\documentclass[a4paper,onecolumn,unpublished,11pt]{quantumarticle}
\pdfoutput=1
\usepackage{silence}
\usepackage{array}
\usepackage{amsthm}
\usepackage{amsmath}
\usepackage{amssymb}
\usepackage{subcaption}
\usepackage{relsize}
\usepackage{tikzit}

\tikzstyle{Node}=[fill=none, draw=black, shape=circle, thick, inner sep=0.75mm]
\tikzstyle{Fixed node}=[fill=none, draw=black, shape=regular polygon, thick, inner sep=0.6mm, regular polygon sides=4]
\tikzstyle{Latent node}=[fill=none, draw=red, shape=circle, thick, inner sep=0.9mm, text=red]
\tikzstyle{new style 0}=[fill=red, draw=red, shape=circle, inner sep=0.5mm]
\tikzstyle{discard_up}=[fill=white, draw=black, shape=ground IEC, rotate=90]
\tikzstyle{discard_down}=[fill=white, draw=black, shape=ground IEC, rotate=-90]
\tikzstyle{label_l}=[fill=none, draw=none, shape=circle, text=red]
\tikzstyle{faded_latent}=[fill=none, draw={rgb,255: red,255; green,191; blue,191}, shape=circle, thick, inner sep=0.9mm, text={rgb,255: red,255; green,191; blue,191}]

\tikzstyle{Edge}=[draw=blue, ->, thick]
\tikzstyle{Bidirected edge}=[<->, draw=red, thick]
\tikzstyle{Unobserved edge}=[->, draw=red, thick]
\tikzstyle{Undirected}=[-, thick]
\tikzstyle{Dotted}=[->, draw=blue, thick, dashed]
\tikzstyle{Arrow}=[->, thick]
\tikzstyle{Outline}=[-, dashed, draw=black]
\tikzstyle{Thick outline}=[-, dashed, thick]
\tikzstyle{faded_edge}=[draw={rgb,255: red,191; green,191; blue,191}, ->, thick]
\tikzstyle{faded_l}=[draw={rgb,255: red,255; green,191; blue,191}, ->, thick]

\usetikzlibrary{fit, positioning, shapes.geometric}
\usepackage[T1]{fontenc}
\newcommand{\pa}{\text{pa}}
\newcommand{\ch}{\text{ch}}

\newcommand{\an}{\text{an}}
\newcommand{\de}{\text{de}}
\newcommand{\nd}{\text{nd}}

\newcommand{\ind}{\mathrel{\perp\!\!\!\perp}}
\newcommand{\aff}{\text{aff}}
\newcommand{\Span}{\text{span}}

\newcommand{\inc}{\text{inc}}
\newtheorem{theorem}{Theorem}
\newtheorem{lemma}[theorem]{Lemma}
\newtheorem{proposition}[theorem]{Proposition}
\newtheorem{definition}[theorem]{Definition}
\newtheorem{conjecture}{Conjecture}
\newtheorem{corollary}{Corollary}[theorem]
\newtheorem{example}[theorem]{Example}
\newtheorem*{remark}{Remark}
\newtheorem*{no_number_theorem}{Theorem}

\usepackage{graphicx}

\usepackage[
backend=biber,
style=lncs,
sorting=none,
citestyle=numeric-comp
]{biblatex}
\DefineBibliographyStrings{english}{%
  backrefpage = {\uparrow},
  backrefpages = {\uparrow},
}
\usepackage[hidelinks]{hyperref}
\begin{document}

\title{Bounding Classical and Quantum Correlations in Bayesian Networks with Quasiprobabilities}
\author{Paul Becsi}
\affiliation{Quantum Group, Department of Computer Science, University of Oxford}
\orcid{0009-0000-6341-2094}

\author{Matty J. Hoban}
\affiliation{Quantum Group, Department of Computer Science, University of Oxford}
\orcid{0000-0001-9765-0373}

\maketitle

\begin{abstract}
Bell’s theorem reveals that quantum theory is in tension with classical causal reasoning and, in particular, the notion of local causality. This is now understood as a particular example of non-classicality in the study of correlations in (Bayesian) networks with both unobserved and observed nodes, where the correlations are probability distributions over the observed nodes. There is a great
deal of work that aims to understand bounds on quantum and classical correlations in such networks and one approach is to consider outer approximations to the former. Along these lines, we consider quasiprobabilistic models for Bayesian networks, which can be seen as classical models with the caveat that probability distributions involving unobserved nodes are "replaced" with quasiprobabilities that respect normalisation but not positivity. We refer to the set of correlations resulting from these models as the \emph{quasi set}. Such models have a history in the study of Bell-type non-classicality where it has been shown that they can produce all non-signalling correlations. We show a generalisation of this result for a broad class of networks, which motivates a conjecture that the quasi set
recovers the so-called nested Markov model. Our work utilises a connection to tensor network decompositions, which may be of independent interest.
\end{abstract}

\keywords{Quantum Causality \and Quasiprobability Distributions \and Bell scenario \and Bell Inequalities}

\section{Introduction}

Bell's theorem reveals that quantum theory is incompatible with local hidden variable models which can be derived from the principle of local causality  \cite{Bell_EPR}. Local causality can be reflected in the structure of a Bayesian network \cite{Pearl_book,WoodSpekkens}, which is a directed acyclic graph (DAG) where the nodes represent either the experimental (observed) variables or the hidden (latent) variable and the edges are the (allowed) causal influences, as depicted in Fig. \ref{fig:bipartite_bell}. Given this DAG one can derive inequality constraints if all of the nodes represent classical systems: the Bell inequalities. What is particularly appealing about this presentation is that it permits a general framework for the study of quantum correlations. That is, we can ask about quantum and classical correlations in different DAGs with observed and latent nodes. This has resulted in a rich literature at the intersection of quantum theory and classical causal inference \cite{Fritz, Fritz2, mDAGs}. One aspect of this is to derive (inequality) constraints for classical models, and another to find quantum violations of these constraints. There are now practical techniques for the former such as the inflation technique \cite{inflationIntroduction}, and the quantum inflation technique for the latter \cite{quantumInflation}.

\begin{figure}[t]
    \ctikzfig{TikzFigures/Bell_Scenario}
    \caption{The bipartite Bell scenario. This is an example of a directed acylic graph (DAG) where random variables are associated with all nodes and arrows represent cause-effect relationships: if one node is a child of another then the latter can influence the former. In this case, the square vertices $X$ and 
    $Y$ are fixed by interventions (such as measurement choices) and the black circle vertices $A$ and $B$ are observed. The red circle vertex is a latent node and not directly observable, but may influence $A$ and $B$ due to the arrows present.}
    \label{fig:bipartite_bell}
\end{figure}

Finding quantum violations of classical constraints is not only of foundational interest, but could have applications in quantum information such as in communication complexity problems \cite{communicationComplexity}, device-independent randomness expansion \cite{randomnessCertificationBell} and quantum key distribution \cite{keyDistribution}. These applications are clearest in the Bell scenario, but there have also been applications in more general networks represented by DAGs \cite{LeeHoban, randomness}. Despite immense progress, characterising the set of quantum correlations in a general network\footnote{There are some subtleties in defining the quantum set even for the Bell experiment, such as whether infinite-dimensional quantum resources are allowed, or whether shared sources need to come from a product Hilbert space or commuting operators. We point the reader to section 1.3 of \cite{quantumCorrelationsDiscussion} for a more detailed discussion of the history of this issue.} is a difficult task, even undecidable in general \cite{quantum_undecidable}. Given the difficulty of studying the quantum set exactly, it has often proved useful to study simpler relaxations of it.

Along these lines, it is well established that quantum theory is just one example of a class of generalised probabilistic theories (GPTs) \cite{Barrett_GPTs}. These GPTs permit violations of Bell inequalities that exceed those permitted by quantum theory; the Popescu-Rohrlich (PR) box is a notable example \cite{pr-box}. Given this observation, Henson, Lal and Pusey in \cite{hlp} considered a theory-independent approach to studying correlations in Bayesian networks. In this work, they presented a generalisation of Bayesian networks (with latents) that allows one to incorporate elements from any (causal) GPT. That is, latent variables could be replaced by preparations of systems in some (possibly post-quantum) theory. Not only is this of foundational interest, it also led to new methods for finding new scenarios where one could potentially see quantum systems producing non-classical correlations; these are called \textit{interesting} DAGs. Within \cite{hlp} a condition was proposed for identifying such DAGs and there has been recent progress on determining if it captures all interesting DAGs \cite{Shashaank}.

Within the field of classical causal inference itself, the study of marginal probability distributions for DAGs with latent variables has seen much progress \cite{mDAGs, nested_intro, nested_completeness}. Motivated by the \textit{d-separation criterion}, one can describe a couple of outer approximations to the set of classical correlations. These are the \textit{nested Markov} and \textit{ordinary Markov} models, with the former contained in the latter. In \cite{nested_completeness}, it was shown that the set of classical correlations is ``algebraically equivalent" to the nested Markov model. Furthermore, the set of distributions in the nested Markov model contains all distributions that can be produced by GPT systems \cite{zhang2025physicsnestedmarkovmodels}. For the Bell scenario, the nested Markov and ordinary Markov sets both coincide with the set of all non-signalling correlations. To connect back to the work of Henson, Lal and Pusey, they used the ordinary Markov set as part of their criteria for identifying interesting DAGs. 

Regarding Bell's theorem and physics more broadly, \textit{quasiprobabilities}, or negative probabilities, have often been studied. Indeed, the use of quasiprobabilities in physics goes back to 1932, with Wigner proposing the phase-state representation of quantum states \cite{Wigner} and Feynman also highlighted the potential for negative probabilities \cite{Feynman}. In parallel, in theoretical computer science, negative probabilities can be found in the context of affine finite automata \cite{affineAutomata}. For the purposes of a Bell experiment (or Bell scenario), one replaces the probability distribution over hidden variables (and the conditional probabilities for outcomes) with signed real numbers; in this way one can again get correlations beyond both classical and quantum theory. It has been established many times independently that such generalisations of local hidden variable models can maximally violate any Bell inequality and even simulate any non-signalling probability distribution \cite{bell_quasi, Degorre_2009, AbramskyBrandenburger, Cavalcanti_2022}. This naturally begs the question of whether this result generalises to other DAGs, and what such a generalisation should look like. This work addresses the role of quasiprobabilities for Bayesian networks, and their relationship to the aforementioned sets of correlations.

Before we describe the results within this paper, it will be insightful to reconsider the Bell scenario in Fig. \ref{fig:bipartite_bell}, a graph which we will call $\mathcal{B}_2$. In the first instance, classical correlations arise when the latent node is a random variable. A distribution\footnote{When discussing Bell scenarios, it is customary to consider the observed distributions to be conditional over some observed inputs, in our case $X$ and $Y$. In the formalism of \cite{nested}, this is the same as considering the Bell scenario to be a CDAG with fixed variables $X$ and $Y$ and reasoning about the probability kernels $p(ab|xy)$ consistent with it. In said paper, fixed variables are drawn as squares, convention which we kept in Fig. \ref{fig:bipartite_bell}.} $p(ab|xy)$ is in the classical set, which we denote $\mathcal{C}(\mathcal{B}_2)$, if it satisfies the following factorisation condition 
\begin{equation}    \label{eq:bell_factorisation}
p(ab|xy)= \sum_{\lambda\in\Lambda}p(\lambda)\cdot p(a|x,\lambda)\cdot p(b|y,\lambda),
\end{equation}
for some sample space $\Lambda$ and (conditional) probability distributions $p(a|x,\lambda)$, $p(b|y,\lambda)$, and $p(\lambda)$. The conditional dependency results from the parent-child relationship in the graph; vertices $X$, $Y$ and $\Lambda$ are parentless. Any distribution in $\mathcal{C}(\mathcal{B}_2)$ must satisfy the \emph{Clauser-Horne-Shimony-Holt (CHSH) inequality} \cite{chsh}:
\begin{equation}
p(00|00)+p(11|00)+p(00|01)+p(11|01)+p(00|10)+p(11|10)+p(01|11)+p(10|11)\leq 3\label{bell_inequality}
\end{equation}
The quasi set $\mathcal{W}(\mathcal{B}_{2})$ of correlations are those correlations $p(ab|xy)$ that have a description according to \eqref{eq:bell_factorisation}, but all of the probability distributions in the summand are replaced with (possibly negative) real numbers. These quasiprobabilities still satisfy normalisation conditions, i.e. $\sum_{\lambda\in\Lambda}p(\lambda)=1$ and
\begin{equation}
\sum_{a}p(a|x,\lambda)=\sum_{b}p(b|y,\lambda)=1.
\end{equation} 

As mentioned, quantum correlations can reproduce all classical correlations but also violate this inequality, thus $\mathcal{C}(\mathcal{B}_2)\subsetneq\mathcal{Q}(\mathcal{B}_2)$, where $\mathcal{Q}(\mathcal{B}_2)$ is the set of quantum correlations for $\mathcal{B}_2$. If we take $\mathcal{GPT}(\mathcal{B}_2)$ to be the set of correlations achievable by any GPT (including quantum and classical theory), then the PR box implies that $\mathcal{Q}(\mathcal{B}_2)\subsetneq\mathcal{GPT}(\mathcal{B}_2)$. The non-signalling set for $\mathcal{B}_2$, denoted as $\mathcal{N}(\mathcal{B}_{2})$ consists of correlations that respect all of the non-signalling constraints, which are conditional independences that can be "read off" from the observed nodes the DAG in Fig. \ref{fig:bipartite_bell} (which coincides with the aforementioned nested Markov model for $\mathcal{B}_2$). For instance, $A$ is independent of $Y$ conditioned on $X$, which can be translated into the following equality constraints:
\begin{equation}
p(a|xy)=\sum_b p(ab|xy)=\sum_b p(ab|xy')=p(a|x)
\end{equation}
where $y\neq y'$. It is well established that $\mathcal{N}(\mathcal{B}_2)=\mathcal{GPT}(\mathcal{B}_2)$. The central result in \cite{bell_quasi}, which has appeared in other works, is that $\mathcal{N}(\mathcal{B}_2)=\mathcal{W}(\mathcal{B}_2)$. 

\subsection{Contributions}

The main contributions of this paper are generalisations of the results in \cite{bell_quasi} to DAGs beyond Bell scenarios. In that paper it was argued that for the quasi set, the expression $p(\lambda)$ in \eqref{eq:bell_factorisation} could be negative, while the conditional probabilities $p(a|x,\lambda)$ and $p(b|y,\lambda)$ (response functions) are valid probabilities, or vice versa; the same set of observed correlations is obtained in both cases. That is, we can limit the negativity to be either in the distributions over latents or in the response functions, and one still obtains $\mathcal{N}(\mathcal{B}_2)$. Firstly, we prove that this observation generalises to any DAG $\mathcal{G}$: one can always limit the negativity to being in the quasiprobabilities over the latents, or in the response functions without loss of generality. Secondly, we will study whether this quasi set contains the set of all non-signalling distributions in general, i.e. whether it coincides with the nested Markov model for any given graph. We will conjecture that this is indeed true for any DAG.

To expand on this conjecture, we will give background to what is known currently and will naturally generalise the discussion above. Given a DAG $\mathcal{G}$ describing a Bayesian network, let $\mathcal{C}(\mathcal{G})$, $\mathcal{Q}(\mathcal{G})$, $\mathcal{GPT}(\mathcal{G})$ and $\mathcal{W}(\mathcal{G})$ be the sets of classical, quantum, GPT\footnote{The union of set of correlations over all GPTs.} and quasiprobabilistic correlations.
It should be clear that $\mathcal{C}(\mathcal{G})\subseteq \mathcal Q(\mathcal{G}) \subseteq \mathcal {GPT}(\mathcal{G})$. Furthermore, define $\mathcal N(\mathcal{G})$ to be the nested Markov model for $\mathcal{G}$\footnote{Here we are actually referring to the nested Markov model of the \emph{latent projection} of $\mathcal{G}$ \cite{nested}, but for simplicity we won't put much emphasis on this distinction.}. Henson, Lal and Pusey showed (implicitly) that $\mathcal{GPT}(\mathcal{G})$ is \textit{strictly} contained in $\mathcal N(\mathcal{G})$ for certain DAGs $\mathcal{G}$ \cite{hlp}. We can ask the following questions:
\begin{itemize}
    \item Is $\mathcal {GPT}(\mathcal{G})\subseteq \mathcal W(\mathcal{G})$?
    \item Is $\mathcal W(\mathcal{G})\subseteq \mathcal N(\mathcal{G})$?
    \item Is $\mathcal N(\mathcal{G})\subseteq \mathcal W(\mathcal{G})$?
\end{itemize}
The first question can be answered positively by appealing to frameworks formalising GPTs, such as \cite{Barrett_GPTs}, where states and measurements are represented by elements of real vector spaces (although we won't do so explicitly in this work). One can also show that the second question has a positive response (and we give an argument in this paper). Finally, we conjecture that the answer to the third question is also yes. Our work makes partial progress on resolving this conjecture. In particular, we provide a proof in the case of correlation scenarios with no bidirected cycles, i.e. the underlying graph has a tree structure.

\subsection{Paper Structure}
We start by introducing graphical models and some results in the causality literature in Section \ref{sec:graphical_models}. The next two sections then contain our main results: in Section \ref{sec:quasiprobability_in_causality} we talk about defining the quasi set and show that all definitions we introduce are equivalent. In Section \ref{sec:quasiprobability_and_the_nested_markov_model} we conjecture that the quasi set coincides with the nested Markov model for all DAGs, then we introduce tensor network decompositions and show how they are related to directed graphical models in the context of correlation scenarios. The section ends with a proof of the conjecture for tree-structured correlation scenarios. Finally, Section \ref{sec:conclusions} contains some conclusions and future work. Most of the results in the main text are stated without proof, with the proofs deferred to the Appendix.

\section{Directed Graphical Models} \label{sec:graphical_models}
A directed graph $\mathcal{G}$ consists of a set of vertices $V$ together with a set of edges $E\subseteq V\times V$, where edges from a vertex to itself are not allowed. Most of the time we follow the convention of using lowercase letters for vertices and uppercase letters for sets of vertices and will denote the edge $(v,w)$ by $v\rightarrow w$. A directed path consists of a sequence of edges respecting direction, i.e. $v_1\rightarrow v_2, v_2\rightarrow v_3,...,v_{n-1}\rightarrow v_n$. A directed cycle is a directed path together with an additional edge $v_n\rightarrow v_1$. A directed acyclic graph (DAG) is a directed graph that has no directed cycles. Often we will write the set of vertices of $\mathcal{G}$ as $V\dot\cup L$, to make clear the graph has a set of observed nodes $V$ and a set of latent nodes $L$.

As usual, given a vertex $v$, we denote by $\pa_\mathcal{G}(v)=\{w\in V:w\rightarrow v\in E\}$ and $\ch_\mathcal{G}(v)=\{w\in V:v\rightarrow w\in E\}$ the sets of parents and children of $v$. We also define the sets $$\an^0_\mathcal{G}(v) = \de^0_\mathcal{G}(v) = \{v\},$$ $$\an^{n+1}_\mathcal{G}(v) = \{w\in V:\exists v'\in \an^n_\mathcal{G}(v).\,w\rightarrow v'\in E\},$$ $$\de^{n+1}_\mathcal{G}(v) = \{w\in V:\exists v'\in de^n_\mathcal{G}(v).\,v'\rightarrow w\in E\}$$ and finally the sets of ancestors, descendants, and nondescendants of $v$ as $\an_\mathcal{G}(v)=\bigcup_{n=0}^\infty \an^n_{\mathcal{G}}(v)$, $\de_\mathcal{G}(v)=\bigcup_{n=0}^\infty \de^n_{\mathcal{G}}(v)$, and $\nd_\mathcal{G}(v)=V\setminus \de_\mathcal{G}(v)$. We define all of the above sets for sets of vertices disjunctively, i.e. $\pa_\mathcal{G}(A)=\bigcup_{v\in A}\pa_{\mathcal{G}}(v)$. Often we will drop the subscript if the graph is clear from context.

When talking about random variables, $X_v$ will denote the random variable corresponding to vertex $v$, whereas $x_v$ will denote a value in its sample space $\mathcal{X}_v$. Similarly, $X_A$ and $x_A$ will denote vectors of these quantities indexed by the elements of set $A$, while $\mathcal{X}_A$ will denote the product space $\prod_{v\in A}\mathcal{X}_v$. For a distribution $p$ over $X_V$ we will write $p(x_V)$ for the probability of $X_V$ taking value $x_V$ under $p$, and will write $X_A\ind X_B\,|\, X_C\,\,\,[p]$ to mean the variables $X_A$ and $X_B$ are independent given $X_C$ in $p$. Sometimes if we want to make explicit that $p$ is defined on some vertex set $V$, we might refer to the distribution by $p(x_V)$; whether this notation refers to a number denoting an entry of $p$ or the whole distribution itself should hopefully be clear from context.

A DAG $\mathcal{G}$, some sample spaces for its observed vertices $\{\mathcal{X}_v\}_{v\in V}$, and a consistency condition generate a graphical model: a set of probability distributions which are said to be "consistent with the graph", or "explained by the graph", or "generated by the graph". Customarily, the model will contain distributions defined over the set of observed nodes only. If no distinction between observed or latent nodes is made, the convention is the entire graph is observed. Often we will just assume the sample spaces of the observed variables are implicit, so we might omit them when defining models. Also, we will mostly work with finite observed variables, and all product spaces are assumed by default.

The most common consistency condition found in literature is the \emph{factorisation criterion}. We say a distribution $p$ \emph{factorises with respect to} DAG $\mathcal{G}$ with vertices $V\dot\cup L$ if the following equation holds
\begin{equation}
    p(x_V)=\int_{x_L}\prod_{v\in V\cup L} p(x_v|x_{\pa(v)})\,dx_L,\,\,\forall x_V\in\mathcal{X}_V,
\end{equation}
for some sample spaces $\{\mathcal{X}_l\}_{l\in L}$ and (conditional) densities $\{p(x_v|x_{\pa(v)})\}_{v\in V\cup L}$. 

We will call the graphical model given by the factorisation criterion \emph{the classical set} (sometimes found in literature as \emph{the marginal model}) and will denote it by $\mathcal{C}(\mathcal{G})$. The classical set is well-studied and fairly well-understood, although not very well-behaved in general. In the case the graph $\mathcal{G}$ is fully observed, the classical set consists of the distributions which are \emph{Markov with respect to} $\mathcal{G}$, namely the distributions satisfying the three equivalent \emph{Markov properties}. See Chapter $3.2.2$ in \cite{Lauritzen_book} for details.

A few results in literature allow us to follow some simplifying assumptions in the context of the classical set. Results in \cite{mDAGs} (namely lemmas $3.7$, $3.8$, and $3.9$) state that any DAG $\mathcal{G}$ with latent variables has the same classical set as the DAG $\tilde{\mathcal{G}}$ obtained by:
\begin{itemize}
    \item[$\bullet$] exogenising all latent variables, i.e. removing all edges that have arrowheads at that variable and replacing them with arrows going from each parent of the latent to each child of the latent;\\
    \item[$\bullet$] eliminating any latent variables whose children are a subset of the children of another latent;\\
    \item[$\bullet$] eliminating any latent variables that have at most one child.
\end{itemize}
Additionally, a result in \cite{universal_bound} states that $p$ factorises with respect to a DAG $\mathcal{G}$ if and only if there is a factorisation of $p$ which has finite sample spaces for the latent variables. Using these results, we can assume WLOG that any latent in a DAG $\mathcal{G}$ is finite discrete, has no parents and at least two children, and its set of children is not a subset of the children of another latent. For such a graph, which we will call a \emph{reduced graph}, the factorisation criterion becomes
\begin{equation}    \label{eq:factorisation}
    p(x_V)=\sum_{x_L}\prod_{v\in V}p(x_v|x_{\pa(v)})\prod_{l\in L}p(x_l),
\end{equation}
for some finite discrete sample spaces $\{\mathcal{X}_l\}_{l\in L}$, (conditional) distributions $\{p(x_v|x_{\pa(v)})\}_{v\in V}$, and distributions $\{p(x_l)\}_{l\in L}$. We will sometimes refer to the conditional distributions $\{p(x_v|x_{\pa(v)})\}_{v\in V}$ as \emph{response functions} (although we are not assuming they are deterministic).

By results in \cite{universal_bound}, the classical set is known to be semialgebraic in the context of finite observed variables: a union of sets, each constrained by some polynomial equalities and inequalities. By relaxing the inequality constraints we obtain the two models we mentioned before: the \emph{ordinary Markov model} $\mathcal{M}(\mathcal{G})$\footnote{This model is denoted by $\mathcal{I}$ in \cite{hlp}} \cite{markov_properties_admg} and the \emph{nested Markov model} \cite{nested_intro, nested}, which we will denote $\mathcal{N}(\mathcal{G})$. It is shown in \cite{nested} (and also in \cite{nested_completeness} more succinctly) that $\mathcal{C}(\mathcal{G})\subseteq\mathcal{N}(\mathcal{G})$. Once again, it is worth mentioning that when we talk about $\mathcal{N}(\mathcal{G})$, we are actually referring to the nested Markov model of the \emph{latent projection} of $\mathcal{G}$. We are not putting too much emphasis on this distinction to avoid going into too many details not relevant to the discussion in this paper.

Another example of a graphical model is the one given in \cite{hlp}, corresponding to the set of distributions which can be obtained from a graph under a certain GPT. For example, \emph{the quantum set} $\mathcal{Q}(\mathcal{G})$ is the set of distributions $p$ which are generalised Markov with respect to $\mathcal{G}$ for quantum theory (see \cite{hlp} for precise definition). The quantum set will contain the classical set for any graph (simply use quantum states to simulate the classical latent variables). In \cite{hlp}, the authors also define the model consisting of all distributions $p$ which are generalised Markov for \emph{some} GPT with respect to a graph, which we will denote $\mathcal{GPT}(\mathcal{G})$ (originally denoted $\mathcal{G}$ in \cite{hlp}). This model will contain the quantum set by definition, and is shown in \cite{gpt_in_nested} to be contained in the nested Markov model.

In short, our notations for common models found in literature are:
\begin{itemize}
    \item[$\bullet$] $\mathcal{C}(\mathcal{G})$: the classical set; in the statistics literature this is sometimes called the marginal model and denoted differently (commonly $\mathcal{M}(\mathcal{G})$ or a variation);\\
    \item[$\bullet$] $\mathcal{Q}(\mathcal{G})$: the quantum set in \cite{hlp};\\
    \item[$\bullet$] $\mathcal{GPT}(\mathcal{G})$: the GPT set, or the set denoted $\mathcal{G}$ in \cite{hlp};\\
    \item[$\bullet$] $\mathcal{N}(\mathcal{G})$: the nested Markov model \cite{nested, nested_completeness} consisting of distributions satisfying all polynomial equalities;\\
    \item[$\bullet$] $\mathcal{M}(\mathcal{G})$: the ordinary Markov model consisting of all distributions satisfying the d-separation independences, or the set $\mathcal{I}$ in \cite{hlp}.
\end{itemize}
The inclusion relationships which hold for any DAG $\mathcal{G}$ are
\begin{equation}
    \mathcal{C}(\mathcal{G})\subseteq
    \mathcal{Q}(\mathcal{G})\subseteq
    \mathcal{GPT}(\mathcal{G})\subseteq
    \mathcal{N}(\mathcal{G})\subseteq
    \mathcal{M}(\mathcal{G}).
\end{equation}

\section{Quasiprobability in Causality} \label{sec:quasiprobability_in_causality}
In this section, we will define the \emph{quasimarginal model}, or the \emph{quasi set}. We will give three definitions of the model and prove they are all equivalent. We want to compare this set with the other models we defined, so we will also assume that observed distributions, even if they are obtained quasiprobabilistically, have entries in $[0,1]$.

\begin{remark}
    We could potentially work with even more general spaces than just $\mathbb{R}$. However, given that we expect real probabilities to have maximum possible expressivity, it seems generalising too far wouldn't provide any new interesting insights. Nonetheless, perhaps treating the case of complex numbers would help with some algebraic geometry arguments, given the algebraic closure of $\mathbb{C}$. Although we don't prove this in this paper, all of our results should generalise to extensions of $\mathbb{R}$ as well.
\end{remark}

We define a \emph{quasiprobability distribution} to be simply a probability distribution with the relaxation that it can take any values in $\mathbb{R}$ instead of being limited to $[0,1]$. Conditional probability and independence are defined similarly. This will suffice for our purposes, but a more detailed treatment is included in Appendix \ref{sec:quasi_formalisation}. We will only work with the case of finite discrete quasiprobabilistic random variables, and we will see later this does not lose generality under the assumption that observed variables are finite.

Our quasimarginal model will consist of the nonnegative observed distributions which are marginals of a full quasidistribution that factorises with respect to the graph. We distinguish two particular types of such full quasidistributions.

\begin{definition}[Distributions with Quasilatents and Quasiresponses]
    Consider a DAG $\mathcal{G}$ with vertices $V\dot{\cup}L$ and quasidistribution $p$ over all its vertices. We say $p$ has \emph{quasilatents} if $p$ can be written as
    $$p(x_{V\cup L}) = \prod_{v\in V\cup L} p(x_v|x_{\pa_\mathcal{G}(v)}),$$ for some \emph{nonnegative} distributions $p(x_v|x_{\pa_\mathcal{G}(v)})$ for all $v\in V$ and quasidistributions $p(x_l|x_{\pa_\mathcal{G}(l)})$ for all $l\in L$. We say $p$ has \emph{quasiresponses} if $p$ can be written in the same form, but this time with $p(x_v|x_{\pa_\mathcal{G}(v)})$ being general quasidistributions, and $p(x_l|x_{\pa_\mathcal{G}(l)})$ being nonnegative.
\end{definition}

The first definition makes sense from a historical point of view, given that this is the definition often used when analysing negative probabilities in the context of the Bell scenario. The second one is motivated by the results in \cite{bell_quasi}. We will show these two definitions are equivalent from the point of view of their marginals over $V$, illustrating a quasiprobabilistic duality between latent variables and response functions.

It is worth mentioning that the first definition seemingly follows the same idea as the HLP formalism \cite{hlp}, namely that only latents are non-classical. One might argue however that even in the HLP framework, quantum measurements are inherently a part of Quantum Theory, so if we were to extend the HLP intuition to quasiprobability, we should therefore allow both latents and response functions to be negative at the same time. This will turn out to give another model identical to the first two, and we will indeed opt for this HLP-esque definition in the Section \ref{sec:quasiprobability_and_the_nested_markov_model} results.

We define a third model below, which will be useful in proofs.

\begin{definition}[Distributions with Quasinoise]
    Consider a DAG $\mathcal{G}$ with vertices $V\dot{\cup} L$ and quasidistribution $p$ over all its vertices. Then we say $p$ has \emph{quasinoise} if there exist random variables $N_v$ for each $v\in V$ (independent of everything in the graph except $X_v$), called noise variables, such that $p$ can be written as:
    $$p(x_{V\cup L}) = \sum_{n_v}\prod_{v\in V} p(x_v|x_{\pa_\mathcal{G}(v)}, n_v)\prod_{l\in L} p(x_l|x_{\pa_\mathcal{G}(l)})\prod_{v\in V} p(n_v),$$ for some deterministic distributions $p(x_v|x_{\pa_\mathcal{G}(v)}, n_v)$, nonnegative distributions $p(x_l|x_{\pa_\mathcal{G}(l)})$ and quasidistributions $p(n_v)$.
\end{definition}

\begin{figure}[t]
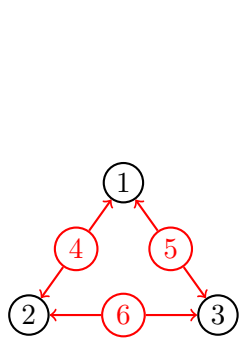
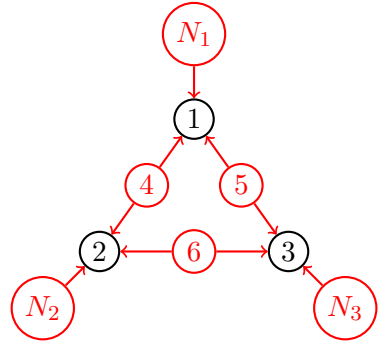

    \begin{subfigure}{0.5\textwidth}
        \ctikzfig{TikzFigures/TriangleDAG}
        \caption{The Triangle Scenario.}
        \label{fig:triangleDAG}
    \end{subfigure}
    \begin{subfigure}{0.5\textwidth}
        \ctikzfig{TikzFigures/TriangleDAGNoise}
        \caption{The Triangle Scenario with noise variables.}
        \label{fig:triangle_noise}
    \end{subfigure}
    \caption{Introducing noise variables, illustrated on the Triangle Scenario.}
    \label{fig:triangle_noise_example}
\end{figure}
Intuitively, noise variables simply isolate away the internal randomness of the observed nodes. See Figure \ref{fig:triangle_noise} for an example. They can be considered finite discrete in the case of finite discrete observed variables\footnote{To see this, note that one can regard noise variables as simply some additional latent variables. Since the construction in Appendix \ref{sec:universal_bound} does not change the response distributions of the observed variables when creating the finite model, one can replace the noise variables with finite counterparts.}.

We are now ready to define the quasimarginal set. We will give three different definitions in terms of the type of full quasidistribution we consider.

\begin{definition}[The Three Quasimarginal Sets]
    Consider DAG $\mathcal{G}$, and nonnegative distribution $p$ over its observed vertices. Then we define the sets
    \begin{itemize}
        \item[$\bullet$] $\mathcal{W}_l(\mathcal{G})$ to be the set of distributions $p$, such that $p$ is the marginal of a quasidistribution with quasilatents that factorises with respect to $\mathcal{G}$;\\
        \item[$\bullet$] $\mathcal{W}_r(\mathcal{G})$ to be the set of distributions $p$, such that $p$ is the marginal of a quasidistribution with quasiresponses that factorises with respect to $\mathcal{G}$;\\
        \item[$\bullet$] $\mathcal{W}_n(\mathcal{G})$ to be the set of distributions $p$, such that $p$ is the marginal of a quasidistribution with quasinoise that factorises with respect to $\mathcal{G}$.
    \end{itemize}
\end{definition}

In the following result we will prove these sets coincide for any DAG $\mathcal{G}$, and hence we can define the quasimarginal set as any of them.

\begin{theorem}[Equivalence of Quasimarginal Sets]  \label{equivalence_theorem}
    For any DAG $\mathcal{G}$, the following holds:
    $$\mathcal{W}_l(\mathcal{G})=\mathcal{W}_r(\mathcal{G})=\mathcal{W}_n(\mathcal{G}).$$
\end{theorem}

We give a full proof of this result in Appendix \ref{sec:equivalence_theorem}. We first assume the DAG $\mathcal{G}$ is reduced and the latents are finite discrete, since the results in \cite{mDAGs} and \cite{universal_bound} extend for quasiprobability distributions as well, which is shown in Appendices \ref{sec:exogenisation} and \ref{sec:universal_bound}. The proof then involves showing $\mathcal{W}_n=\mathcal{W}_r$ and $\mathcal{W}_n=\mathcal{W}_l$. The former is fairly simple. For the latter, the $\mathcal{W}_n\subseteq\mathcal{W}_l$ direction involves just encapsulating any noise variable into a latent parent of its observed child. The $\mathcal{W}_l\subseteq\mathcal{W}_n$ direction is the more involved step of the proof, and involves a lemma that shows a quasiprobabilistic latent parent can be replaced by a classical latent parent and quasiprobabilistic noise variables for its children. This process is then applied for each latent.

It is not too difficult to see that Theorem \ref{equivalence_theorem} can be extended to also include the set of distributions which are marginals of some general factorising quasidistribution (i.e. full quasidistributions where neither observed nor latent response functions are restricted to $[0,1]$).

We can therefore define the quasi-marginal set as follows:
\begin{definition}[The Quasimarginal Set]\label{definition_quasimarginal_set}
    Given any DAG $\mathcal{G}$, we define the quasimarginal set $\mathcal{W}(\mathcal{G})$ to be any of the sets $\mathcal{W}_l(\mathcal{G})$, $\mathcal{W}_r(\mathcal{G})$, $\mathcal{W}_n(\mathcal{G})$. Equivalently, it is equal to the set of observed distributions $p$ which can be written as 
    $$p(x_V) =\sum_{x_L} \prod_{v\in V} p(x_v|x_{\pa_\mathcal{G}(v)})\prod_{l\in L}p(x_l),$$
    for some (conditional) quasidistributions $p(x_v|x_{\pa(v)})$ and $p(x_l)$.
\end{definition}

\section{Quasiprobability and the Nested Markov Model}  \label{sec:quasiprobability_and_the_nested_markov_model}
We will now discuss the relationship between the quasimarginal set and the nested Markov model.

\begin{proposition} \label{proposition:quasi_is_in_nested}
    For any DAG $\mathcal{G}$,
    $$\mathcal{W}(\mathcal{G})\subseteq\mathcal{N}(\mathcal{G}).$$
\end{proposition}

\begin{proof}
    Follows from the fact that the nested Markov model is defined entirely in terms of the factorisation criterion, which is obeyed by distributions in the quasi set (albeit with quasiprobabilistic factors). Indeed, the treatment in \cite{nested} is purely symbolic and does not rely on positivity (nor realness, in fact) of the distributions. A more detailed discussion can be found in Appendix \ref{sec:proof_quasi_is_in_nested}.
\end{proof}

We conjecture that, even more strongly, equality holds for any graph.

\begin{conjecture}  \label{conjecture}
    For any DAG $\mathcal{G}$,
    $$\mathcal{W}(\mathcal{G})=\mathcal{N}(\mathcal{G}).$$
\end{conjecture}

We will now study the conjecture in the context of correlation scenarios \cite{Fritz} with a tree structure.

\begin{definition}[Correlation Scenario]
    A DAG $\mathcal{G}$ is a \emph{correlation scenario} if it is reduced and observed variables have no children.
\end{definition}

Therefore, a correlation scenario $\mathcal{G}$ can be seen as a directed bipartite graph: it consists of a layer of latent variables and a layer of observed variables, and all edges are from some latent variable to some observed one. We will also assume the graphs we are working with are connected graphs (otherwise just separate the problem into each connected component). The Markov model and the nested Markov model have particularly simple forms in the context of correlation scenarios.

\begin{proposition} \label{proposition:correlation_scenario_properties}
    Let $\mathcal{G}$ be a correlation scenario. Then a distribution $p$ is in the Markov model of $\mathcal{G}$ iff it satisfies $A\ind B\,\,[p]$ whenever $A$ and $B$ are two subsets of $V$ such that no latent variable has children in both sets. Additionally, $\mathcal{N}(\mathcal{G})=\mathcal{M}(\mathcal{G})$.
\end{proposition}

For the first claim, see for example \cite{markov_properties_admg} (Theorem 3) for a proof in the case of graphs with only binary latents, which directly generalises. The second claim follows from the definition of the nested Markov property in \cite{nested_completeness} and an inductive argument.

We will assume from now on that all correlation scenarios we are working with only have latent variables with exactly two children. Less than that would make the latent variable redundant, whereas three children or more can be reduced to the simpler case via the following proposition.

\begin{proposition} \label{proposition:correlation_scenarios_with_2_children_reduction}
    Let $\mathcal{G}$ be a correlation scenario with a latent $\Lambda$ with $k\geq 3$ children. Consider now the correlation scenario $\mathcal{G}'$ which replaces $\Lambda$ with $\binom{k}{2}$ latents with two children each, one for every pair of children of $\Lambda$. Then $\mathcal{M}(\mathcal{G'})=\mathcal{M}(\mathcal{G})$ and $\mathcal{W}(\mathcal{G'})\subseteq\mathcal{W}(\mathcal{G})$.
\end{proposition}

The first result follows the fact that the two graphs have the same latent projection, which encodes all observed independences. For the second result, if $p$ can be obtained by a latent assignment in $\mathcal{G}'$, then simply make $\Lambda$ in $\mathcal G$ jointly simulate all of its corresponding latent variables in $\mathcal{G}'$, and make its children just pick out the information they have access to in $\mathcal{G}'$. This shows $p$ has a latent variable construction in $\mathcal{G}$ as well.

\begin{corollary}
    It is enough to prove Conjecture \ref{conjecture} for correlation scenarios where all latent variables have exactly two children.
\end{corollary}
\begin{proof}
    Assume we have proven Conjecture \ref{conjecture} for correlation scenarios where all latent variables have exactly $2$ children and let $\mathcal{G}$ be an arbitrary correlation scenario. Consider the graph $\mathcal{G}'$ defined in Proposition \ref{proposition:correlation_scenarios_with_2_children_reduction}. Then by Propositions \ref{proposition:quasi_is_in_nested}, \ref{proposition:correlation_scenario_properties}, and \ref{proposition:correlation_scenarios_with_2_children_reduction} we have $$\mathcal{M}(\mathcal{G})\supseteq\mathcal{W}(\mathcal{G})\supseteq\mathcal{W}(\mathcal{G}')=\mathcal{M}(\mathcal{G}')=\mathcal{M}(\mathcal{G}).$$
\end{proof}

Now, we will discuss a certain subclass of these correlation scenarios.

\begin{definition}
    Let $\mathcal{G}$ be a correlation scenario. Consider the undirected graph $\mathcal{G}^t$ that has set of vertices $V$ (i.e. the observed vertices of $\mathcal{G}$), and an edge between two vertices if they have a common latent parent in $\mathcal{G}$. We say $\mathcal{G}$ is a \emph{tree-structured correlation scenario (TCS)} if $\mathcal{G}^t$ is a tree.
\end{definition}

Note that if $\mathcal{G}$ is a TCS, then every latent variable must have exactly $2$ children: if $\Lambda$ has at least $3$ children, then any two of them will be connected in $\mathcal{G}^t$, introducing a cycle. As we mentioned before, we will prove Conjecture \ref{conjecture} in the context of TCSs. One important feature that allows this is the very simple structure of the ordinary Markov model (and hence the nested Markov model too) of TCSs.

\begin{proposition} \label{markov_model_in_TCS}
    Let $\mathcal{G}$ be a TCS. For each vertex $v$, note that removing $v$ from the graph separates it into a bunch of connected components. We will call the set of such components $\mathcal{B}_v$, meaning $$\mathcal{B}_v=\{B\subseteq{V\setminus \{v\}}:B\text{ is a connected component of the graph }\mathcal{G}\setminus\{v\}\}.$$

    Then $p\in\mathcal{M}(\mathcal{G})$ iff for any vertex $v$ and $B\in\mathcal{B}_v$, we have $$B\ind V\setminus(B\cup\{v\})\,\,[p],$$
    or equivalently, $$p(x_{V\setminus\{v\}})=\prod_{B\in\mathcal{B}_v} p(x_B).$$ 
\end{proposition}

In other words, all the (nested) Markov model requires is that after eliminating a vertex $v$, the distribution with $x_v$ summed out factorises with respect to the branches left.

Next, we introduce tensor network decompositions and then we will provide a proof for Conjecture \ref{conjecture} for TCSs.

\subsection{Tensor Network Decompositions}
Tensor network decompositions are a tool that allows storing the information of a tensor more compactly. They are used in machine learning and studying structures of entanglement in multipartite quantum states such as in quantum many-body physics, but it was noted they are also related to graphical models \cite{GraphicalModelsAndTensors}. They involve decomposing a tensor into an expression consisting of tensor products and contractions following the structure of a graph, where each vertex is assigned a separate tensor, and edges represent contraction of the neighbouring components over an index of fixed size, called a \emph{bond dimension}. For now we assume the entries of the tensors come from an arbitrary field $k$, and will replace $k:=\mathbb{R}$ later when we get back to graphical models. For an undirected graph $\mathcal{G}$, we will denote by $\inc_\mathcal{G}(v):=\{(v_1,v_2)\in E:v_1=v\text{ or }v_2=v\}$ the set of incident edges of vertex $v$, and will drop the index when the graph is clear from context.

We start with the formal definition of a tensor network decomposition and then we will give an example.

\begin{definition}[Tensor Network Decomposition]    \label{def:tensor_network_decompositions}
    Let $\mathcal{G}$ be an undirected graph. Assign to each vertex $v\in V$ a finite-dimensional $k$-vector space $\mathbb{V}_v$ and each edge $e\in E$ a bond dimension $r_e\in\mathbb{Z}_+$. We say a tensor $\mathcal{T}\in\bigotimes_{v\in V} \mathbb{V}_v$ \emph{decomposes with respect to $\mathcal{G}$ with bond dimensions $\{r_e\}_{e\in E}$} if there exist $k$-vector spaces $\{\mathbb{E}_e\}_{e\in E}$ of corresponding dimensions $\{r_e\}_{e\in E}$ and tensors $\left\{T^{(v)}\in\mathbb{V}_v\otimes\left( \bigotimes_{e\in\inc(v)}\mathbb{E}_e\right)\right\}_{v\in V}$ such that 
    \begin{equation}    \label{eq:tensor_factorisation_initial_definition}
    \mathcal{T}(x_V)=\sum_{\{1\leq i_e\leq r_e\}_{e\in E}}\prod_{v\in V} T^{(v)}\left(x_v,\{i_e\}_{e\in \inc(v)}\right).
    \end{equation}
\end{definition}

Note that the sum in \eqref{eq:tensor_factorisation_initial_definition} is taken over multiple indices, i.e. it can be read as $|E|$ separate nested sums from $i_e=1$ to $r_e$, for each edge $e\in E$. Also, for simplicity of notation, from now on when referring to entries of $T^{(v)}$, we will move the indices $i_e$ of the edges incident to $v$ to the subscript, so equation \eqref{eq:tensor_factorisation_initial_definition} becomes:
\begin{equation}    \label{eq:tensor_factorisation}
    \mathcal{T}(x_V)=\sum_{\{1\leq i_e\leq r_e\}_{e\in E}}\prod_{v\in V} T^{(v)}_{\{i_e\}_{e\in \inc(v)}}(x_v).
    \end{equation}

\begin{figure}[t]
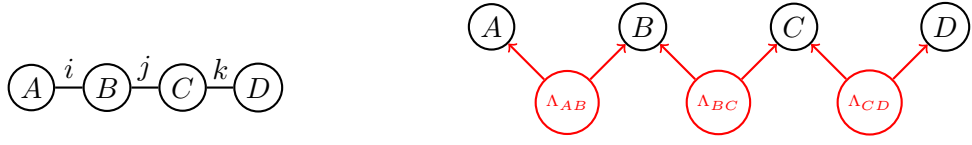

    \begin{subfigure}{0.5\textwidth}
        \ctikzfig{TikzFigures/4-on-line_undirected}
        \caption{A tree $\mathcal{T}$ with edges labeled by indices.}
        \label{fig:4-on-line_undirected}
    \end{subfigure}
    \begin{subfigure}{0.5\textwidth}
        \ctikzfig{TikzFigures/4-on-lineDAG}
        \caption{The corresponding TCS $\mathcal{T}^{DAG}$ of $\mathcal{T}$.}
        \label{fig:4-on-lineDAG}
    \end{subfigure}

    \caption{The correspondence between the $4$-on-line Scenario and its underlying undirected graph. Edges become latent variables with the same domain size as the corresponding index.}
\end{figure}

\begin{example} \label{example_tree_tensor_decomposition}
Let's look at the undirected graph\footnote{Note that often in quantum literature, one would draw open edges at each vertex in order to emphasise the index corresponding to $v$ of the tensor $T^{(v)}$.} $\mathcal{T}$ in Figure \ref{fig:4-on-line_undirected}, and the TCS $\mathcal{T}^{DAG}$ in \ref{fig:4-on-lineDAG} that we obtain by replacing every edge in $\mathcal{T}$ by a latent variable, which we will call the \emph{$4$-on-line Scenario}. For some finite sets $\mathcal{A}, \mathcal{B}, \mathcal{C}, \mathcal{D}$ we consider the order-$4$ real-valued tensors $T(a,b,c,d)$, where $a\in\mathcal{A}$, $b\in\mathcal{B}$, $c\in\mathcal{C}$, $d\in\mathcal{D}$. We say such a tensor $T$ decomposes with respect to $\mathcal{T}$ and bond dimensions $I,J,K\in\mathbb{Z}_+$ if there exist tensors
\begin{gather*}
   A\in\mathbb{R}^{|\mathcal{A}|}\,\otimes\,\mathbb{R}^I, \quad B\in\mathbb{R}^{|\mathcal{B}|}\,\otimes\,\mathbb{R}^I\otimes \,\mathbb{R}^J, \quad 
   C\in\mathbb{R}^{|\mathcal{C}|}\,\otimes\,\mathbb{R}^J\otimes\, \mathbb{R}^K,\\
    D\in\mathbb{R}^{|\mathcal{D}|}\,\otimes\,\mathbb{R}^K
\end{gather*}
such that 
\begin{equation}    \label{eq:ex_decomposition}
    T(a,b,c,d)=\sum_{i,j,k} A_i(a) B_{ij}(b)C_{jk}(c)D_k(d).
\end{equation}

Now assume we have a distribution $p(a,b,c,d)\in\mathcal{W}\left(\mathcal{T}^{DAG}\right)$, where the random variables $X_A,X_B,X_C,X_D$ take values in the sets $\mathcal{A}$, $\mathcal{B}$, $\mathcal{C}$, $\mathcal{D}$. In other words, there exist latent variables $\Lambda_{AB}$, $\Lambda_{BC}$, $\Lambda_{CD}$, such that 
\begin{align*}    p(a,b,c,d)=&\sum_{\lambda_{AB},\lambda_{BC},\lambda_{CD}} p(a|\lambda_{AB})\cdot p(b|\lambda_{AB},\lambda_{BC})\cdot p(c|\lambda_{BC},\lambda_{CD})\\
&\cdot p(d|\lambda_{CD})\cdot p(\lambda_{AB})\cdot p(\lambda_{BC})\cdot p(\lambda_{CD}),
\end{align*}
for some value assignments to the (conditional) quasidistributions in the expression. Given this expression, we can equivalently write
\begin{align*}
p(a,b,c,d)=&\sum_{\lambda_{AB},\lambda_{BC},\lambda_{CD}} p(a,\lambda_{AB})\cdot p(b|\lambda_{AB},\lambda_{BC})\cdot p(c,\lambda_{BC}|\lambda_{CD})\cdot p(d,\lambda_{CD}),
\end{align*}
by simply grouping factors, for example $p(a|\lambda_{AB})\cdot p(\lambda_{AB})=p(a,\lambda_{AB})$. Therefore, if we now define \begin{gather*}
    A_{\lambda_{AB}}(a):= p(a,\lambda_{AB}),\quad 
    B_{\lambda_{AB},\lambda_{BC}}(b)=p(b|\lambda_{AB},\lambda_{BC}),\quad\\
    C_{\lambda_{BC},\lambda_{CD}}(c):=p(c,\lambda_{BC}|\lambda_{CD}),\quad
    D_{\lambda_{CD}}(d):=p(d,\lambda_{CD}),
\end{gather*}
and view $p(a,b,c,d)$ as an order-$4$ tensor, we actually recover equation \eqref{eq:ex_decomposition}, where the indices now range over the values of the latent variables. Note that we could've instead grouped $p(\lambda_{AB})$ into the factor corresponding to $b$ via $p(b|\lambda_{AB},\lambda_{BC})\cdot p(\lambda_{AB})=p(b,\lambda_{AB}|\lambda_{BC})$ to obtain a different tensor network decomposition. This shows the process we applied is not unique.

In general, given a latent variable representation of a distribution with respect to a correlation scenario $\mathcal{G}$, we can always construct a tensor network decomposition with respect to the corresponding undirected graph $\mathcal{G}^t$ by following the same process, where the bond dimensions are the sizes of the latents. Our proof of Conjecture \ref{conjecture} will explore going in the opposite direction.
\end{example}

A pretty common result in the literature of tensor network states is that every tensor can be decomposed with respect to a graph $\mathcal{G}$ as long as we allow the bond dimensions to be large enough (the result is more common for trees \cite{TreeTensorNetworks, TreeTensorNetworks2}; it is proven in the general case in \cite{TensorNetworkRanks}). It is also known that in the case of a tree graph, a \emph{minimal} decomposition exists for every tensor \cite{TensorNetworkRanks, TreeTensorNetworks}.

We now have all the ingredients we need for our proof of Conjecture \ref{conjecture} in the case of TCSs.

\begin{theorem}
    For any Tree-structured Correlation Scenario $\mathcal{G}$, a distribution is in the quasimarginal model if and only if it satisfies the observed independence constraints enforced by the graph. In other words, $$\mathcal{W}(\mathcal{G})=\mathcal{N}(\mathcal{G})=\mathcal{M}(\mathcal{G}).$$   \label{theorem:TCS_reals}
\end{theorem}

\begin{proof}
It follows from Propositions \ref{proposition:quasi_is_in_nested} and \ref{proposition:correlation_scenario_properties} that $\mathcal{W}(\mathcal{G})\subseteq\mathcal{N}(\mathcal{G})=\mathcal{M}(\mathcal{G})$. For the opposite inclusion, consider an arbitrary element of the ordinary Markov model $p\in\mathcal{M}(\mathcal{G})$. We will build a quasiprobabilistic latent variable representation of $p$.

We will view $p(x_V)$ as a tensor: it has order $|V|$, one mode (i.e. degree of freedom) for each vertex of the graph (or equivalently observed variable). Now consider the undirected graph $\mathcal{G}^t(V, E^t)$ which has as vertices the observed vertices of $\mathcal{G}$ and an edge between $v$ and $w$ iff $v$ and $w$ have a common (latent) parent. By the results in \cite{TensorNetworkRanks} (more precisely Theorem 8.3), there exist minimal bond dimensions $(r_1,\dots,r_{|E^t|})$ such that $p(x_V)$ has a tensor network decomposition with respect to $\mathcal{G}^t$. The tuple $(r_1,\dots,r_{|E^t|})$ is called the \emph{$\mathcal{G}^t$-rank} of $p$, and its minimality is with respect to the product partial order on tuples. We will transform this minimal decomposition of $p(x_V)$ to show that $p(x_V)\in\mathcal{W}(\mathcal{G})$. We will illustrate the proof for the case of the $4$-on-line Scenario in Fig. \ref{fig:4-on-lineDAG} and direct the reader to Appendix \ref{sec:proof_conjecture_trees} for the general case.

In the case of the $4$-on-line Scenario, the minimal tensor decomposition corresponds to minimal $I,J,K\in\mathbb{Z_+}$ such that there exist tensors
\begin{gather*}
   A\in\mathbb{R}^{|\mathcal{A}|}\,\otimes\,\mathbb{R}^I, \quad B\in\mathbb{R}^{|\mathcal{B}|}\,\otimes\,\mathbb{R}^I\otimes \,\mathbb{R}^J, \quad 
   C\in\mathbb{R}^{|\mathcal{C}|}\,\otimes\,\mathbb{R}^J\otimes\, \mathbb{R}^K,\\
    D\in\mathbb{R}^{|\mathcal{D}|}\,\otimes\,\mathbb{R}^K
\end{gather*}
that satisfy the equation
\begin{equation}    \label{eq:4-on-line_decomposition}
    p(a,b,c,d)=\sum_{i,j,k} A_i(a) B_{ij}(b)C_{jk}(c)D_k(d).
\end{equation}

The first step is to note that WLOG the following inequations hold:
\begin{gather*}
    \sum_a A_i(a)\neq0\\
    \sum_b B_{ij}(b)\neq 0\\
    \sum_c C_{jk}(c)\neq 0\\
    \sum_d D_k(d)\neq 0,
\end{gather*}
for all values $1\leq i,j,k\leq I,J,K$, result which we call the \emph{Nonzero Sums Lemma}; the proof can be found in sections \ref{sec:multilinear_products}, \ref{sec:gauge_transformations}, and \ref{sec:nonzero_assumption} of the Appendix. Intuitively, if the decomposition we are given doesn't satisfy this, then one can always construct one that does using the \emph{gauge transformations} defined in \ref{sec:gauge_transformations}.

Next, notice that we can write:
\begin{align*}
    p(a,b,c,d)
    &=\sum_{i,j,k} A_i(a) B_{ij}(b)C_{jk}(c)D_k(d)\\
    &=\sum_{i,j,k}\frac{A_i(a)}{\sum_a A_i(a)}\cdot\frac{B_{ij}(b)}{\sum_b B_{ij}(b)}\cdot\frac{C_{jk}(c)}{\sum_c C_{jk}(c)}\cdot\frac{D_k(d)}{\sum_d D_k(d)}\\
    &\cdot\left(\sum_a A_i(a)\right)\left(\sum_b B_{ij}(b)\right)\left(\sum_c C_{jk}(c)\right)\left(\sum_d D_k(d)\right),
\end{align*}
and we have $\sum_a\frac{A_i(a)}{\sum_a A_i(a)}=1,\forall i$, meaning we can consider $\frac{A_i(a)}{\sum_a A_i(a)}$ to be $p(a|\lambda_{AB}=i)$, and the same for the other observed nodes. What we need now is for the factor at the end of the expression to decompose as $$\left(\sum_a A_i(a)\right)\left(\sum_b B_{ij}(b)\right)\left(\sum_c C_{jk}(c)\right)\left(\sum_d D_k(d)\right)=\phi(i)\cdot\psi(j)\cdot\theta(k).$$ Indeed, if this holds then we can write 
\begin{align*}
    1
    &=\sum_{a,b,c,d}p(a,b,c,d)\\
    &=\sum_{i,j,k}\left(\sum_a\frac{A_i(a)}{\sum_a A_i(a)}\right)\cdot\left(\sum_b\frac{B_{ij}(b)}{\sum_b B_{ij}(b)}\right)\cdot\left(\sum_c\frac{C_{jk}(c)}{\sum_c C_{jk}(c)}\right)\cdot\left(\sum_d\frac{D_k(d)}{\sum_d D_k(d)}\right)\\
    &\cdot\left(\sum_a A_i(a)\right)\left(\sum_b B_{ij}(b)\right)\left(\sum_c C_{jk}(c)\right)\left(\sum_d D_k(d)\right)\\
    &=\sum_{i,j,k}1\cdot 1\cdot1\cdot1\cdot\phi(i)\cdot\psi(j)\cdot\theta(k)\\
    &=\sum_i\phi(i)\cdot\sum_j\psi(j)\cdot\sum_k\theta(k),
\end{align*}
and hence 
\begin{align*}
    p(a,b,c,d)
    &=\sum_{i,j,k} A_i(a) B_{ij}(b)C_{jk}(c)D_k(d)\\
    &=\sum_{i,j,k}\frac{A_i(a)}{\sum_a A_i(a)}\cdot\frac{B_{ij}(b)}{\sum_b B_{ij}(b)}\cdot\frac{C_{jk}(c)}{\sum_c C_{jk}(c)}\cdot\frac{D_k(d)}{\sum_d D_k(d)}\\
    &\cdot\frac{\phi(i)}{\sum_i\phi(i)}\cdot\frac{\psi(j)}{\sum_j\psi(j)}\cdot\frac{\theta(k)}{\sum_k\theta(k)}.
\end{align*}
Then we would have $\sum_i\frac{\phi(i)}{\sum_i\phi(i)}=1$ and the same for $j$ and $k$, so by setting
\begin{gather*}
p(\lambda_{AB}=i):=\frac{\phi(i)}{\sum_i\phi(i)},\\
p(\lambda_{BC}=j):=\frac{\psi(j)}{\sum_j\psi(j)},\\
p(\lambda_{CD}=k):=\frac{\theta(k)}{\sum_k\theta(k)},
\end{gather*}
we would be done.

It is therefore enough to show that $\sum_b B_{ij}(b)$ and $\sum_c C_{jk}(c)$ factorise into components only depending on $i$ and $j$, and $j$ and $k$ respectively. This is where we will use the independence constraints of the ordinary Markov model. We will show that $\sum_b B_{ij}(b)$ decomposes and the same will follow for $C$. The proof just involves some reasoning about matrices, and the general proof in Appendix \ref{sec:proof_conjecture_trees} is a tensor extension of this argument.

We start by summing \eqref{eq:4-on-line_decomposition} over $b$:
\begin{equation}
    p(a,c,d)=\sum_{i,j,k}A_i(a)\cdot\left(\sum_b B_{ij}(b)\right)\cdot C_{jk}(c)\cdot D_k(d).  \label{eq:summed_decomposition}
\end{equation}
We now rearrange the expression to act on order-two objects, i.e. matrices. The LHS $p(a,c,d)$ is an order-$3$ tensor, but we can rewrite its entries into a $|\mathcal{A}|\times(|\mathcal{C}|\cdot|\mathcal{D}|)$ matrix with rows corresponding to $a$, and columns corresponding to $(c,d)$. We will denote this matrix by $p(a;(c,d))$. Similarly, the tensor $\sum_k C_{jk}(c)\cdot D_k(d)$ has three dimensions: $j$, $c$, and $d$, but we can rewrite it as a matrix with dimensions corresponding to $j$ and $(c,d)$, matrix which we will denote by $M$. Hence if we write $A$ as a $|\mathcal{A}|\times I$ matrix, $\sum_b B_{ij}(b)$ as a $I\times J$ matrix, and $M$ as a $J\times(|\mathcal{C}|\cdot|\mathcal{D}|)$ matrix, then equation \eqref{eq:summed_decomposition} becomes:
\begin{equation}
    p(a;(c,d))= A \left(\sum_b B_{ij}(b)\right)M.    \label{eq:flattened_matrix_decomposition}
\end{equation}

The idea now is to note that $A$ and $M$ have inverses on one side each. Indeed, we claim that the columns of the $|\mathcal{A}|\times I$ matrix $A$ are linearly independent. Assume for the sake of contradiction that there exists $1\leq q\leq I$ such that $$A_q(a)=\sum_{i\neq q}\alpha_i A_i(a),$$ for some $\{\alpha_i\in\mathbb{R}\}_{i\neq q}$. Then going back to \eqref{eq:4-on-line_decomposition}, we get:
\begin{align*}
    p(a,b,c,d) 
    &=\sum_{i,j,k} A_i(a) B_{ij}(b)C_{jk}(c)D_k(d)\\
    &=\sum_{i\neq q,j,k}A_i(a) B_{ij}(b)C_{jk}(c)D_k(d) + \sum_{j,k}A_q(a) B_{qj}(b)C_{jk}(c)D_k(d)\\
    &=\sum_{i\neq q,j,k}A_i(a) B_{ij}(b)C_{jk}(c)D_k(d) + \sum_{i\neq q,j,k} \left(\alpha_i A_i(a)\right)B_{qj}(b)C_{jk}(c)D_k(d)\\
    &=\sum_{i\neq q,j,k}A_i(a)\cdot\left(B_{ij}(b)+\alpha_i B_{qj}(b)\right)\cdot C_{jk}(c) D_k(d),
\end{align*}
which is a tensor decomposition of $p$ with bond dimensions $(I-1,J,K)$, contradicting the fact that our initial bond dimensions were minimal. Therefore, the columns of $A$ are linearly independent, meaning $A$ has a left inverse $A^L$ such that $A^L A=\mathcal{I}_I$, where $\mathcal{I}_I$ is the $I\times I$ identity matrix. Similarly (but perhaps less obvious; see Appendix \ref{sec:actual_conjecture_proof}), the matrix $M$ has a right inverse $M^R$ such that $MM^R=\mathcal{I}_J$. Multiplying to the left by $A^L$ and to the right by $M^R$ in equation \eqref{eq:flattened_matrix_decomposition} we get: $$A^L\cdot p(a;(c,d))\cdot M^R= \sum_b B_{ij}(b).$$

But $p\in\mathcal{M}(\mathcal{G})$, meaning $p(a,c,d)=p(a)p(c,d)$, or equivalently $p(a;(c,d))=\mathbf{p(a)}\otimes \mathbf{p(c,d)}^T$, where we consider $\mathbf{p(a)}$ and $\mathbf{p(c,d)}$ to be column vectors (and hence we write them in bold). Therefore, we get $$\left(A^L\cdot \mathbf{p(a)}\right)\otimes\left(\mathbf{p(c,d)}^T\cdot M^R\right)= \sum_b B_{ij}(b).$$

Finally, if we look at the types in the LHS, $A^L$ is a $I\times |\mathcal{A}|$ matrix, meaning $A^L\cdot p(a)^T$ is a column vector of size $I$, and similarly $p(c,d)^T\cdot M^R$ is a row vector of size $J$. Therefore, $\sum_b B_{ij}(b)$ indeed factorises into an $i$ component and a $j$ component. The conclusion follows.

\end{proof}

\begin{corollary}   \label{corollary_directed_models_tensor_network_decompositions_equivalence}
    Let $\mathcal{G}$ be a TCS. Then for any distribution $p\in\mathcal{M}(\mathcal{G})$, $p$ has a tensor network decomposition with respect to $\mathcal{G}^t$ and bond dimensions $\{r_e\}_{e\in E^t}$ if and only if it has a quasiprobabilistic latent variable construction where the latent corresponding to edge $e$ in $\mathcal{G}^t$ has domain size equal to $r_e$.
\end{corollary}
\begin{proof}
    Let $\{r_e^*\}_{e\in E^t}$ be the $\mathcal{G}^t$-rank of $p$. Then we know by \cite{TensorNetworkRanks} that $p$ has a decomposition with bond dimensions $\{r_e\}_{e\in E^t}$ if and only if $r_e\geq r_e^*$ for all $e$.

    Now assume $p$ has a tensor network decomposition with bond dimensions $\{r_e\}_{e\in E^t}$. Then $r_e\geq r_e^*$ for all $e$, and $p$ has a tensor decomposition with bond dimensions $\{r_e^*\}_{e\in E^t}$. By the discussion in Appendix \ref{sec:proof_conjecture_trees}, it has a latent variable construction with sizes $\{r_e^*\}_{e\in E^t}$. Since $r_e\geq r_e^*$ for all $e$, this implies the same must be true for latent variable sizes $\{r_e\}_{e\in E^t}$.
    
    Conversely, if $p$ has a latent variable construction with sizes $\{r_e\}_{e\in E^t}$, then it has a tensor network decomposition with bond dimensions $\{r_e\}_{e\in E^t}$ by the construction we have given in the Example \ref{example_tree_tensor_decomposition}.
\end{proof}

Corollary \ref{corollary_directed_models_tensor_network_decompositions_equivalence} therefore illustrates a sort of equivalence between tensor network decompositions and quasiprobabilistic latent variable constructions in the case of non-signalling distributions and TCSs. Note, however, that this doesn't mean any tensor network decomposition of $p$ can be converted to a latent variable factorisation following the construction in Theorem \ref{theorem:TCS_reals}: the minimality of the tensor network rank is crucial for the process to work. The corollary simply states that for any tensor network decomposition of \emph{certain size}, there exists a latent variable construction of the same size, but it might come from a different process than the one in the theorem.

\section{Conclusions and discussion}   \label{sec:conclusions}
In this work, we have formalised the set of quasidistributions consistent with a graph, generalising some known results about the bipartite Bell scenario. We have shown a duality between latent variables and response functions from a quasiprobabilistic point of view. Finally, we proved the quasi set is the same as the nested Markov model for any correlation scenario with a tree structure.

Clearly, the main open problem from our work is to establish whether $\mathcal N(\mathcal{G})\subseteq \mathcal W(\mathcal{G})$ is true for all graphs. In addition to the family of graphs mentioned above, we include in Appendix \ref{sec:triangle_proof} a proof of the conjecture for the well-studied triangle scenario (Fig. \ref{fig:triangleDAG}), or $\mathcal{C}_{3}$ \cite{Fritz}. This graph is a cyclic correlation scenario, suggesting the result in Theorem \ref{theorem:TCS_reals} is not restricted to the tree-structured settings we discussed. The graph $\mathcal{C}_3$ is known to be a small example of \emph{interesting network}: the set $\mathcal{N}(\mathcal{C}_{3})$ is equal to all possible probability distributions on three random variables, but in \cite{hlp}, it was shown that $\mathcal{GPT}(\mathcal{C}_{3})$ does not contain all such distributions. This illustrates that the triangle network is far from trivial, so it is interesting that our hypothesis holds for it as well.

A proof of Conjecture \ref{conjecture} would provide an affirmative answer to Proposition 3.4 of \cite{Fritz}: it is known that GPTs are not expressive enough to generate all distributions in the nested Markov model, not even if we restrict to correlation scenarios, as the perfect correlation in the Triangle Scenario is not in $\mathcal{GPT}$ \cite{hlp}. However, if the conjecture is true, then the theory allowing latents (or equivalently response functions or both) to take all real-valued probabilities would be an example of a theory that can generate any non-signalling distribution without ever being able to violate non-signalling.

It is interesting to ask what it is about quasiprobabilistic models that gives them so much expressive power. One answer to this appears in \cite{computationalLandscape}, where a quasiprobabilistic Turing Machine model was shown to simulate (efficiently) computation in all GPTs. This quasiprobabilistic Turing Machine was transformed into a probabilistic theory, which was dubbed a \textit{non-free} GPT: this means that circuits cannot be arbitrarily composed and there are some constraints on which systems can compose as part of a larger set of systems (or circuit). To give some intuition, we can imagine that a transformation may be applied to part of a system preparation device, but if that device is replaced with another, it could result in negative \textit{observed} probabilities, which is prohibited; the solution is to forbid the latter composition of devices in the theory. Those GPTs that can be freely composed are termed \textit{free} GPTs, and they include the classical, quantum, and boxworld (i.e. PR-box) theories, but don't include quasiprobability theory. This leads to an interesting interaction between quasiprobability and the inflation technique, where the argument in \cite{inflationIntroduction} showing the perfect correlation is not GPT-consistent with the triangle scenario (by invoking the so-called \emph{cut inflation}) breaks down for quasiprobabilities precisely because of this non-freeness property: a quasiprobabilistic inflated distribution exists for the perfect correlation\footnote{In fact, we believe we have a proof for the following statement: for any $n\geq 2$ and correlation scenario $\mathcal{G}$, an observed distribution $p$ has a quasiprobabilistic $n$-th order inflation (defined in \cite{inflationCompleteness}) if and only if $p\in\mathcal{N}(\mathcal{G})=\mathcal{M}(\mathcal{G})$.}, which of course, makes sense given our Appendix \ref{sec:triangle_proof}. In contrast to the above, in \cite{Cavalcanti24}, it was shown how to transform a particular family of quasiprobabilistic models into free GPTs. We leave it to future work to reconcile this seemingly-free GPT interpretation of quasiprobabilities with the GPT-quasi gap of the triangle scenario.

\section*{Acknowledgements:}
We would like to thank Jonathan Barrett, Robin Evans, Tein van der Lugt, Marina Maciel Ansanelli, M\'{a}t\'{e} Weisz, and Ruiwen Dong for the useful discussions that helped build the results in this paper, as well as Christopher Chang for proofreading.

\section*{Author contributions:}
Both authors were involved in the research and writing of this work. PB contributed with the connection to tensor network decompositions, the technical proofs, and wrote most of the manuscript. MJH contributed with the idea behind Theorem \ref{equivalence_theorem}, some input for the technical proofs, wrote the introduction, and provided some editing to the paper.

\printbibliography[
heading=bibintoc, 
title={References}
]

\renewcommand{\theHsection}{A\arabic{section}}

\appendix

\section{Quasiprobability Definitions} \label{sec:quasi_formalisation}
\begin{definition}[Quasiprobability Distribution]
    Given a sample space $\Omega$ and $\sigma$-algebra $\mathcal{F}$, a \emph{quasiprobability measure} or \emph{quasimeasure} is a function $p:\mathcal{F}\rightarrow \mathbb{R}$ satisfying countable disjoint additivity and that $p(\Omega)=1$. The triple $(\Omega,\mathcal{F},p)$ is a \emph{quasiprobability space}. If the space $\Omega$ is finite or countable, we call $p$ a \emph{quasiprobability distribution} or \emph{quasidistribution}.
\end{definition}

\begin{definition}[Conditional Quasidistribution and Independence]
    Given quasidistribution $p$ and event $B$ such that $p(B)\neq 0$, we define the conditional distribution $p(A|B)$ as $$p(A|B)=\frac{p(A,B)}{p(B)}.$$

    We say events $A$ and $B$ are \emph{independent given} $C$ with $p(C)\neq 0$ if $$p(A,B|C)=p(A|C)\cdot p(B|C).$$
\end{definition}

\begin{definition}[Quasiprobabilistic Random Variable]\
    Given quasiprobability space $(\Omega,\mathcal{F},p)$, a random variable is a function $X:\Omega\rightarrow E$, where $E$ is a measurable space.
\end{definition}

\section{Extensions of Classical Results for Quasiprobability}
\label{sec:extensions_classical_results}

\subsection{Soundness of Reducing a Graph}    \label{sec:exogenisation}
In this section we will show the results in \cite{mDAGs} hold for our quasiprobabilistic models as well. However, we cannot assume yet that the three definitions of the quasiset are equivalent (and indeed, we want to use the results in this section to prove the equivalence). Therefore, we would in theory need to prove the three results in \cite{mDAGs} for each of the three definitions. For the sake of brevity we will omit the measure theoretic approach in the original paper. If full rigourosity is desired, these proofs could be rewritten in a more formal way by following the original measure theoretical arguments more closely.

\begin{definition}[Exogenisation]
Consider a DAG $\mathcal{G}$ and a vertex $u$. Then the \emph{exogenised graph} $\mathfrak{r}(\mathcal{G},u)$ is defined as follows: take the parents of $u$ in $\mathcal{G}$ and draw an arrow from each of them to each child of $u$, and erase the arrows pointing to $u$.
\end{definition}

\begin{proposition}[Quasilatents Exogenisation - Lemma $3.7$ in \cite{mDAGs}]   \label{proposition:exogenisation}
Consider DAG $\mathcal{G}$ with vertices $V\dot\cup \{u\}$. Then $p(x_V)$ is the marginal of a quasidistribution with quasilatents/quasiresponses/quasinoise Markov wrt $\mathcal{G}$ iff the same is true for $\mathfrak{r}(\mathcal{G},u)$.
\end{proposition}
\begin{proof}
For the 'if' direction, assume $p(x_V)$ is the marginal of a quasidistribution with quasilatents/quasiresponses/quasinoise Markov wrt $\mathfrak{r}(\mathcal{G},u)$. Then one can create a quasidistribution of the same type consistent with $\mathcal{G}$ that marginalises to $p$ by simply setting $u_{new}:=(u,x_{\pa_\mathcal{G}(u)})$, i.e. make the new $u$ simply behave the same as before, but also send down the values of its parents. If only latents/responses/local noise variables were quasiprobabilistic before, the same is true for the new construction.

For the 'only if' direction, assume $p(x_V)$ is the marginal of a quasidistribution with quasilatents/quasiresponses/quasinoise Markov wrt $\mathcal{G}$. Extract the intrinsic randomness of the variable corresponding to $u$ into a noise variable $e_u$, i.e. assume $u=f_u(e_u,x_{\pa(u)})$ for some deterministic function $f_u$. Then, we can create a quasidistribution consistent with $\mathfrak{r}(\mathcal{G},u)$ that marginalises to $p$ as follows: make $u_{new}:=e_u$ and make each child of $u$ apply the function $f_u$ to simulate the value of $u$. Each such variable has all the ingredients to do so: it is now a child of all the original parents of $u$, as well as its noise variable. Once again, this constructs the same type of quasidistribution (i.e. with quasilatents/quasiresponses/quasinoise) as before.
\end{proof}

\begin{proposition}[Elimination of Redundant Quasilatents - Lemma $3.8$]    \label{proposition:redundant_quasilatents}
Let $\mathcal{G}$ be a DAG with vertices $V\dot\cup\{u,w\}$ such that $\pa(u)=\pa(w)=\emptyset$ and $\ch(w)\subseteq\ch(u)$. Then $p(x_V)$ is the marginal of a quasidistribution with quasilatents/quasiresponses/quasinoise Markov wrt $\mathcal{G}$ iff the same is true for the graph $\mathcal{G}_{-w}$ obtained by removing $w$.
\end{proposition}
\begin{proof}
The 'if' direction is trivial. For the 'only if' direction, make $u_{new}$ simulate both $u$ and $w$ independently and delete $w$. All nodes in $\ch(u)\setminus\ch(w)$ can simply ignore the $w$-component of $u_{new}$. Trivially, the marginal of this construction over $V$ is the same, and the obtained quasidistribution is of the same type.
\end{proof}

\begin{proposition}[Elimination of Latent Variables with One Child - Lemma $3.9$]   \label{proposition:one_child}
Let $\mathcal{G}$ be a DAG with vertices $V\dot\cup\{u\}$ such that $u$ has no parents and at most one child. Then $p(x_V)$ is the marginal of a quasidistribution with quasilatents/quasiresponses/quasinoise Markov with respect to $\mathcal{G}$ iff the same is true for $\mathcal{G}_{-u}$.
\end{proposition}
\begin{proof}
The 'if' direction is trivial. For the 'only if' direction, we can assume WLOG that $u$ has a child $v$ (otherwise the case is trivial). We can also assume WLOG that $v$ has no other unobserved parents. Otherwise simply apply Proposition \ref{proposition:redundant_quasilatents} to eliminate $u$ as a redundant quasilatent. If $u$ is nonnegative, then it can be eliminated using the original version of this result in \cite{mDAGs} (alternatively, see the proof of $\mathcal{W}_n\subseteq\mathcal{W}_r$ in Appendix \ref{sec:equivalence_theorem}). If $u$ is quasiprobabilistic, we can actually replace $u$ with a nonnegative variable and keep the same marginal over $V$ (see the beginning of the case $\mathcal{W}_n\subseteq\mathcal{W}_l$ in Appendix \ref{sec:equivalence_theorem}), then eliminate it. In any case, the same type of quasidistribution is obtained.
\end{proof}

\subsection{Sufficiency of Finite Discrete Latent Variables}    \label{sec:universal_bound}
In this section we will extend one of the main results in \cite{universal_bound} for Quasiprobability. We will use the same measure theoretic notation as the cited paper. Note that the proof requires only that the latents in the graph are exogenised, i.e. the proof holds for reduced graphs.

\begin{proposition}[Finiteness of Quasilatents]
Given DAG $\mathcal{G}$ with vertices $V\dot\cup L$ such that all latents are parentless, a distribution $p(x_V)$ over finite discrete variables $X_V$ can be written as $$p(x_V)=\int_{\Omega_L}\prod_{v\in V}p(x_v|x_{\pa(v)})\prod_{l\in L} d\rho_l(x_l),$$ for (quasi or not) distributions $p(x_v|x_{\pa(v)})$ for $v\in V$ and quasiprobability measures $\rho_l(x_l)$ over sets $\Omega_l$, if and only if it can be written in the same form but with all latents finite discrete.
\end{proposition}
\begin{proof}
Our proof will follow the one in \cite{universal_bound}. We will replace each latent, in turn, with a finite discrete variable over the same space (or equivalently a finite subset of it, i.e. the image of the finite variable), keeping the observed response distributions the same.

Pick a latent variable $l\in L$. We will perform the following experiment: Pick an arbitrary value $x_l^*\in\Omega_l$, fix $X_l$ to deterministically take value $x_l^*$ and then construct $p$ using the factorisation formula as before. The behaviour of the resulting model is:
$$p_{x_l^*}(x_V)=\left[\int_{\Omega_{L\setminus \{l\}}}\prod_{v\in V}p(x_v|x_{\pa(v)})\prod_{l'\in {L\setminus \{l\}}} d\rho_{l'}(x_{l'})\right]_{x_l=x_l^*}.$$

By the square brackets around the right-hand side we mean every occurence of $x_l$ in the expression is replaced by the value $x_l^*$. Now, the vector $\overrightarrow{p_{x_l^*}}$ indexed by $x_l^*$ is a random variable that depends on $x_l^*$ with average $$\mathbb{E}[\overrightarrow{p_{x_l^*}}]=\int_{\Omega_l}\overrightarrow{p_{x_l^*}} \,\,d\rho_l(x_l^*)=p(x_V).$$

Note that $p\in \mathbb{R}^d$, where $d=|\mathcal{X}_V|$. We can look at the subset $U=\{p_{x_l^*}|x_l^*\in\Omega_l\}$ of $\mathbb{R}^d$ consisting of the observed distributions resulting from the experiment over the entire range of possible deterministic choices for $x_l^*\in\Omega_l$. By construction (and the fact that $l$ has no parents), point $p(x_V)$ is an affine mixture of points in $U$. Then by Lemma \ref{lemma_caratheodory}, $p(x_V)$ is a finite affine combination with weights in $\mathbb{R}$ of points in $U$, i.e. we can write $$p(x_V)=\sum_{x_l}\int_{\Omega_{L\setminus \{l\}}}\prod_{v\in V}p(x_v|x_{\pa(v)})\prod_{l\in {L\setminus \{l\}}} d\rho_l(x_l).$$

Applying this process to the rest of the latent variables then gives the desired result.
\end{proof}

\begin{lemma}\label{lemma_caratheodory}
Given subset $U\subset \mathbb{R}^d$ and $\overrightarrow{x}^*\in \mathbb{R}^d$ a (possibly continuous) affine mixture of points in $U$, then $\overrightarrow{x}^*$ is a finite affine combination of points in $U$ with weights in $\mathbb{R}$.
\end{lemma}
\begin{proof}
Consider $\rho$ to be a quasiprobability measure on (the Borel subsets of) $\mathbb{R}^d$ with the support set of the mass function being a subset of $U$. We need to prove that $\overrightarrow{x}^*=\int_{U}\overrightarrow{x}\,\,d\rho(\overrightarrow{x})$ is in the affine hull $\aff(U)$. Note that  $\aff(U)=\overrightarrow{u}+\Span(U-\overrightarrow{u})$, where we fix some $\overrightarrow{u}\in U$.

Now $\Span(U-\overrightarrow{u})$ consists of the solutions to a homogeneous linear system of equations. Consider $\overrightarrow{a}\cdot\overrightarrow{x}=a_1x_1+\dots a_dx_d=0$ to be one such equation, where $x_1,\dots x_d$ are the entries of an arbitrary vector $\overrightarrow{x}\in \mathbb{R}^d$. Then, substituting $\overrightarrow{x}^*-\overrightarrow{u}$ in the above we get:
\begin{align*}
\overrightarrow{a}\cdot(\overrightarrow{x}^*-\overrightarrow{u})&=a_1(x^*_1-u_1)+\dots+a_d(x^*_d-u_d)\\
&=(a_1x^*_1+\dots+a_dx^*_d)-(a_1u_1+\dots+a_du_d)\\
&=\left(a_1\int_U x_1\,d\rho(\overrightarrow{x})+\dots+a_d\int_Ux_d\,d\rho(\overrightarrow{x})\right)\\
&\,\,\,\,\,\,\,-(a_1u_1+\dots+a_du_d)\\
&=\int_U (a_1x_1+\dots+a_dx_d)\,d\rho(\overrightarrow{x})-(a_1u_1+\dots+a_du_d)\\
&=(a_1u_1+\dots+a_du_d)\int_U 1d\,\rho(\overrightarrow{x})-(a_1u_1+\dots+a_du_d)\\
&=0,
\end{align*}
where the second-to-last equality comes from the fact that $$a_1x_1+\dots+a_dx_d=a_1u_1+\dots+a_du_d,$$ because $\overrightarrow{x}-\overrightarrow{u}\in \Span(U-\overrightarrow{u})$, and the last equality follows from the fact that $\rho$ is an affine mixture. Indeed, all these calculations will hold for all linear equations in the system, hence $\overrightarrow{x}^*-\overrightarrow{u}$ is a solution to the system, so $\overrightarrow{x}^*\in\aff(U)$.
\end{proof}

\section{Proof of Theorem \ref{equivalence_theorem}}    \label{sec:equivalence_theorem}
We will only prove the result for reduced graphs and finite latents. This then implies the general case. To see this, consider an arbitrary graph $\mathcal{G}$. Then by Appendix \ref{sec:exogenisation} we have that $\mathcal{W}_l(\mathcal{G})=\mathcal{W}_l(\tilde{\mathcal{G}})$, and the same holds for $\mathcal{W}_r$ and $\mathcal{W}_n$, where $\tilde{\mathcal{G}}$ is the graph $\mathcal{G}$ after reducing it by repeatedly applying propositions \ref{proposition:exogenisation}, \ref{proposition:redundant_quasilatents}, and \ref{proposition:one_child}. By Appendix \ref{sec:universal_bound} we can assume the latents of $\tilde{\mathcal{G}}$ are finite discrete; more precisely, if we can obtain $p(x_V)$ with quasilatents, quasiresponses, or quasinoise, then we can do the same thing even if we restrict latents to be finite discrete (the quasilatents and quasinoise cases are covered by our proof in \ref{sec:universal_bound}, and quasiresponses are covered by the original proof in \cite{universal_bound}). Finally, if we have shown the theorem is true for reduced graphs with finite latents, then it is true for $\tilde{\mathcal{G}}$. This will imply $\mathcal{W}_l(\tilde{\mathcal{G}})=\mathcal{W}_r(\tilde{\mathcal{G}})=\mathcal{W}_n(\tilde{\mathcal{G}})$, so the same must be true for $\mathcal G$ as well.

Before the main proof, we start with a few useful definitions.

\begin{definition}[Assignments]
    Given a reduced DAG $\mathcal{G}$, an \emph{assignment} is a set of fixed quasidistributions $p(x_v|x_{\pa_{\mathcal{G}}(v)})$, one for each vertex $v\in V$, and a set of fixed quasidistributions $p(x_l)$, one for each vertex $l\in L$ (as well as the necessary probability spaces for the latents). We say the assignment \emph{achieves} observed distribution $p(x_V)$ if $$p(x_V)=\sum_{x_L}\prod_{v\in V} p(x_v|x_{\pa_{\mathcal{G}}(v)})\prod_{l\in L} p(x_l).$$ Furthermore, we say the assignment:
    \begin{itemize}
        \item[$\bullet$] \emph{witnesses} $p\in \mathcal{W}_l$ if it achieves $p$, and $p(x_v|x_{\pa_{\mathcal{G}}(v)})\in[0,1]$ for all $v\in V$.
        \item[$\bullet$] \emph{witnesses} $p\in \mathcal{W}_r$ if it achieves $p$ and $p(x_l)\in[0,1]$ for all $l\in L$.
    \end{itemize}
\end{definition}

\begin{definition}[Noisy Assignments]
    Given a DAG $\mathcal{G}$, a \emph{noisy assignment} is a set of random variables $\{N_v\}_{v\in V}$ (whose values we will denote by $n_v$ as opposed to $x_{N_v}$; note that each $N_v$ is a single variable, not a set of variables, despite the upper-case letter notation), together with deterministic distributions $p(x_v|x_{\pa_{\mathcal{G}}(v)}, n_v)$ for all $v\in V$, quasiprobabilistic distributions $p(x_l)$ for each latent $l\in L$ of $\mathcal{G}$, and quasiprobabilistic distributions $p(n_v)$ for each $v\in V$. We say this noisy assignment \emph{achieves} $p(x_V)$ if $$p(x_V) = \sum_{x_L,n_V}\,\,\prod_{v\in V} p(x_v|x_{\pa_{\mathcal{G}}(v)}, n_v)\prod_{l\in L} p(x_l)\prod_{v\in V} p(n_v).$$ Furthermore, we say the noisy assignment \emph{witnesses} $p\in \mathcal{W}_n$ if $p(x_l)\in[0,1]$ for all $l\in L$.
\end{definition}

We are now ready to prove the quasimarginal set definition equivalence.

\begin{no_number_theorem}[\textbf{Equivalence of Quasimarginal Sets}]
For any DAG $\mathcal{G}$, the following holds:
    $$\mathcal{W}_l(\mathcal{G})=\mathcal{W}_r(\mathcal{G})=\mathcal{W}_n(\mathcal{G}).$$
\end{no_number_theorem}

\begin{proof}
Fix a DAG $\mathcal{G}$. Start by noting that for any noisy assignment, there exist functions $f_v:\mathcal{X}_{\pa_{\mathcal{G}}(v)}\times \mathcal{X}_{N_v} \rightarrow \mathcal{X}_v$, one for each $v\in V$, defined as $f_v(x_{\pa_{\mathcal{G}}(v)},n_v)=x_v^*$, where $x_v^*\in \mathcal{X}_v$ is the only value such that $p(x_v^*|x_{\pa_\mathcal{G}(v)}, n_v)\neq0$, or equivalently $p(x_v^*|x_{\pa_\mathcal{G}(v)}, n_v)=1$ (since for noisy assignments, these conditionals are deterministic). Therefore, a noisy assignment achieves $p(x_V)$ iff the following holds:
$$p(x_V) = \sum_{\substack{x_L,n_V\\f_v(x_{\pa(v)},n_v)=x_v\,\,\forall \,v\in V}}\prod_{l\in L} p(x_l)\prod_{v\in V} p(n_v).$$
The proof of the theorem consists of showing that $\mathcal{W}_r=\mathcal{W}_n$ and $\mathcal{W}_l=\mathcal{W}_n$.\\\\
$\bullet\,\,\,\mathcal{W}_n\subseteq\mathcal{W}_r$\\
\textit{Idea:} encapsulate the noise variables into the response function of each observed node. Assume $p(x_V)$ is in $\mathcal{W}_n$ and consider a noisy assignment $p^*$ that witnesses this. Then,
\begin{align*}
p(x_V) &= \sum_{x_L,n_V}\,\,\prod_{v\in V} p^*(x_v|x_{\pa_{\mathcal{G}}(v)}, n_v)\prod_{l\in L} p^*(x_l)\prod_{v\in V} p^*(n_v)\\
    &= \sum_{x_L,n_V}\,\,\prod_{v\in V} p^*(x_v, n_v|x_{\pa_{\mathcal{G}}(v)})\prod_{l\in L} p^*(x_l)\\
    &= \sum_{x_L}\prod_{v\in V} \left(\sum_{n_v} p^*(x_v, n_v|x_{\pa_{\mathcal{G}}(v)})\right)\prod_{l\in L} p^*(x_l)\\
    &=\sum_{x_L}\prod_{v\in V} p^*(x_v|x_{\pa_{\mathcal{G}}(v)})\prod_{l\in L} p^*(x_l),
\end{align*}
and hence $p^*$ can be used to construct an assignment that achieves $p$. Additionally, since we assumed that $p^*$ witnesses $p\in\mathcal{W}_n$, we must have $p^*(x_l)\in[0,1]$ for all $l\in L$, therefore we can indeed say the newly constructed assignment witnesses $p\in\mathcal{W}_r$.\\\\

\noindent
$\bullet\,\,\,\mathcal{W}_r\subseteq\mathcal{W}_n$\\
\textit{Idea:} isolate the randomness of response functions into noise variables. Consider $p\in\mathcal{W}_r$ and an assignment $p^*$ that witnesses this. The goal is to create a noise variable for $X_v$ that simulates its distribution conditioned on its parents. This will lead to a noisy assignment witnessing $p\in\mathcal{W}_n$

As such, consider for each $X_v$ the variable $N_v$ given by the tuple of independent variables $$N_v=\left(N_v^{(c_1)}, N_v^{(c_2)},\dots,N_v^{(c_s)}\right),$$ $\{c_1,\dots,c_s\}=\mathcal{X}_{\pa_{\mathcal{G}}(v)}$, i.e. the $c$ quantities are the possible values of the parents of $v$. Note that we can assume WLOG that $|\mathcal{X}_{\pa_\mathcal{G}(v)}|$ is a finite quantity by Appendix \ref{sec:universal_bound}. We choose $$N_v^{(c_i)}\sim p^*(x_v|x_{\pa(x_v)}=c_i).$$

Finally, we define the functions $f_v$ of our noisy assignment as:
$$f_v\left(x_{\pa(v)},n_v=\left(n_v^{(c_1)},\dots,n_v^{(c_s)}\right)\right):=n_v^{(x_{\pa(v)})}.$$ Therefore, $N_v^{(c_i)}$ dictates the behaviour of $X_v$ for one value of its parents. If we consider the noisy assignment $p^{**}$ that has $p^{**}(x_l)=p^*(x_l)$ for all $l\in L$ and noise variables described above, then we have the following for every observed node $v\in V$:
\begin{align*}
p^{**}(x_v|x_{\pa(v)})
    &=\sum_{n_v}p^{**}(x_v|x_{\pa_{\mathcal{G}}(v)},n_v)p^{**}(n_v)\\ 
    &= \sum_{\substack{n_v\\f_v(x_{\pa(v)},n_v)=x_v}} p^{**}(n_v)\\
    &= \sum_{\substack{\left(n_v^{(c_1)},\dots,n_v^{(c_s)}\right)\\f_v\left(x_{\pa(v)},n_v=\left(n_v^{c_1},\dots,n_v^{c_s}\right)\right)=x_v}} p^{**}\left(n_v=\left(n_v^{(c_1)},\dots,n_v^{(c_s)}\right)\right)\\
    &= \sum_{\substack{\left(n_v^{(c_1)},\dots,n_v^{(c_s)}\right)\\n_v^{\left(x_{\pa(v)}\right)}=x_v}} p^{**}\left(n_v^{(c_1)}\right)\dots p^{**}\left(n_v^{(c_s)}\right)\\
    &=p^{**}\left(n_v^{\left(x_{\pa(v)}\right)}=x_v\right)\\
    &=p^*(x_v|x_{\pa(v)}).
\end{align*}

Finally, this implies that:
\begin{align*}
&\sum_{x_L,n_V}\,\prod_{v\in V} p^{**}(x_v|x_{\pa_{\mathcal{G}}(v)}, n_v)\prod_{l\in L} p^{**}(x_l)\prod_{v\in V} p^{**}(n_v)\\
    &\hspace{0.5cm}= \sum_{x_L}\,\prod_{v\in V}\sum_{n_v} p^{**}(x_v|x_{\pa_{\mathcal{G}}(v)}, n_v)p^{**}(n_v)\prod_{l\in L} p^{**}(x_l)\\
    &\hspace{0.5cm}= \sum_{x_L}\,\prod_{v\in V}p^{**}(x_v|x_{\pa_{\mathcal{G}}(v)})\prod_{l\in L} p^{**}(x_l)\\
    &\hspace{0.5cm}= \sum_{x_L}\prod_{v\in V} p^*(x_v|x_{\pa_{\mathcal{G}}(v)})\prod_{l\in L} p^*(x_l)\\
    &\hspace{0.5cm}= p(x_V),
\end{align*}
so indeed $p^{**}$ achieves $p(x_V)$. Since we haven't changed the latent variables, they are still nonnegative, and the response functions are deterministic, so $p^{**}$ is indeed a noisy assignment witnessing $p\in\mathcal{W}_n$.

Therefore, we have shown that $\mathcal{W}_r=\mathcal{W}_n$. We move on to $\mathcal{W}_l=\mathcal{W}_n$ next.\\\\

\noindent
$\bullet\,\,\,\mathcal{W}_n\subseteq\mathcal{W}_l$\\
\textit{Idea:} Encapsulate the noise variables into latent parents where possible, and into response functions where no latent parents are present. Consider $p\in\mathcal{W}_n$ and let $p^*$ be a noisy assignment witnessing that. Note that if an observed vertex $a$ has no latent parents in $\mathcal{G}$, then its noise variable can be assumed nonnegative WLOG. To see this, assume $N_a$ is not nonnegative. Then one can do the construction we used to prove $\mathcal{W}_n\subseteq\mathcal{W}_r$ to create an assignment $p^*_1$ witnessing $p\in\mathcal{W}_r$, and then use the construction from $\mathcal{W}_r\subseteq\mathcal{W}_n$ to create another noisy assignment $p_2^*$ witnessing $p\in\mathcal{W}_n$. However, this new noisy assignment has $$N_a=\left(N_a^{(c_1)},\dots,N_a^{(c_s)}\right),$$ where $p^*_2\left(n_a^{(c_i)}\right)=p_1^*(x_a|x_{\pa_{\mathcal{G}}(a)})$. But since $a$ has no latent parents, the fact that $p_1^*$ achieves $p(x_V)$ implies by Lemma \ref{proof_no_latent_parents} that
\begin{equation}
p_1^*(x_v|x_{\pa_{\mathcal{G}}(v)})=p(x_v|x_{\pa_\mathcal{G}(v)}).\label{nonnegative_noise}
\end{equation}
Therefore, each of the $N_a^{(c_i)}$s has nonnegative distributions, so $N_a$ is nonnegative as well.

We will now iteratively eliminate noise variables from the noisy assignment and have them be simulated by latent variables where possible, or by response variables where no latent parents are present.

More precisely, for an observed variable $a\in V$ with no latent parents,  consider a new assignment $p_{new}^*$, where $N_a$ has been replaced by a trivial variable $N_a^{new}$ with a single possible value $1_a$. We define $$p_{new}^*(x_a|x_{\pa_{\mathcal{G}}(a)},1_a):=\sum_{n_a}p^*(x_a|x_{\pa_{\mathcal{G}}(a)},n_a)p^*(n_a).$$ Otherwise, have $p^*_{new}$ behave exactly the same as $p^*$, and in particular we set $N_v^{new}:=N_v$ for $v\neq a$.

Then we can show that $p^*_{new}$ achieves $p(x_V)$ as well:
\begin{align*}
p(x_V) 
    &=\sum_{x_L,n_V}\,\,\prod_{v\in V} p^*(x_v|x_{\pa_{\mathcal{G}}(v)}, n_v)\prod_{l\in L} p^*(x_l)\prod_{v\in V} p^*(n_v)\\
    &=\sum_{x_L}\left(\sum_{n_a}p^*(x_a|x_{\pa_{\mathcal{G}}(a)}, n_a)p^*(n_a)\right)\cdot\\
    &\,\,\,\,\,\left(\sum_{n_{V\setminus\{a\}}}\,\,\prod_{v\in V\setminus\{a\}} p^*(x_v|x_{\pa_{\mathcal{G}}(v)}, n_v)\prod_{l\in L} p^*(x_l)\prod_{v\in V\setminus\{a\}} p^*(n_v)\right)\\
    &=\sum_{x_L}p_{new}^*(x_a|x_{\pa_{\mathcal{G}}(a)},1_a)\cdot\\
    &\,\,\,\,\,\left(\sum_{n_{V\setminus\{a\}}}\,\,\prod_{v\in V\setminus\{a\}} p^*_{new}(x_v|x_{\pa_{\mathcal{G}}(v)}, n_v)\prod_{l\in L} p^*_{new}(x_l)\prod_{v\in V\setminus\{a\}} p^*_{new}(n_v)\right)\\
    &=\sum_{x_L,n_V^{new}}\,\,\prod_{v\in V} p^*_{new}(x_v|x_{\pa_{\mathcal{G}}(v)}, n_v^{new})\prod_{l\in L} p^*_{new}(x_l)\prod_{v\in V} p^*_{new}(n_v^{new}),
\end{align*}
and the new response function $p^*_{new}(x_a|x_{\pa(a)})$ we gave $a$ is real nonnegative, as a result of (\ref{nonnegative_noise}). Next, replace $p^*$ by $p^*_{new}$ and repeat the process until we trivialise all of the noise variables associated with variables with no latent parents.

We then apply a similar process for the rest of the observed variables, namely the ones with latent parents. For these, the noise variables can be from all of $\mathbb{R}$, so we hide them latents.

More precisely, for each observed variable $a$ that has at least one latent parent $l_a$, construct a new assignment $p^*_{new}$ by replacing $N_a$ with a trivial variable $N_a^{new}$ with a single value $1_a$, replacing $l_a$ by $l_a^{new}=(l_a,N_a)$, and setting $$p^*_{new}\left(x_a|x_{\pa_{\mathcal{G}}(a)\setminus\{l_a\}},x_{l_a^{new}}=(x_{l_a},n_a),1_a\right):=p^*(x_a|x_{\pa_{\mathcal{G}}(a)\setminus\{l_a\}},x_{l_a},n_a).$$ Otherwise, make $p^*_{new}$ behave the same as $p^*$. In particular, for any other observed child $b$ of $l_a$, set $$p^*_{new}\left(x_b|x_{\pa_{\mathcal{G}}(b)\setminus\{l_a\}},x_{l_a^{new}}=(x_{l_a},n_a),n_b\right):=p^*(x_b|x_{\pa_{\mathcal{G}}(b)\setminus\{l_a\}},x_{l_a},n_b),$$
so given that $x_{l_a^{new}}=(x_{l_a},n_a)$, we have that $$p^*_{new}\left(x_b|x^{new}_{\pa_{\mathcal{G}}(b)},n_b\right):=p^*(x_b|x_{\pa_{\mathcal{G}}(b)},n_b).$$

Then, once again, $p^*_{new}$ achieves $p(x_V)$ as well:
\begin{align*}
p(x_V)
    &=\sum_{x_L,n_V}\,\,\prod_{v\in V} p^*(x_v|x_{\pa_{\mathcal{G}}(v)}, n_v)\prod_{l\in L} p^*(x_l)\prod_{v\in V} p^*(n_v)\\
    &=\sum_{x_{L\setminus\{l_a\}},n_{V\setminus\{a\}}}\sum_{n_a,x_{l_a}} \left(p^*(n_a)p^*(x_{l_a})\right)\cdot
    p^*(x_a|x_{\pa_{\mathcal{G}}(a)\setminus\{l_a\}},x_{l_a}, n_a)\cdot\\
    &\,\,\,\,\,\left(\prod_{v\in V\setminus\{a\}} p^*(x_v|x_{\pa_{\mathcal{G}}(v)}, n_v)\prod_{l\in L\setminus\{l_a\}} p^*(x_l)\prod_{v\in V\setminus\{a\}} p^*(n_v)\right)\\
    &=\sum_{x_{L\setminus\{l_a\}},n_{V\setminus\{a\}}}\,\,\sum_{x_{l_a}^{new}=(n_a,x_{l_a})} p^*_{new}(x_{l_a}^{new})\cdot p^*_{new}\left(x_a|x_{\pa_{\mathcal{G}}(a)\setminus\{l_a\}},x_{l_a^{new}},1_a\right)\\
    &\,\,\,\,\,\left(\prod_{v\in V\setminus\{a\}} p^*_{new}(x_v|x^{new}_{\pa_{\mathcal{G}}(v)}, n_v)\prod_{l\in L\setminus\{l_a\}} p^*_{new}(x_l)\prod_{v\in V\setminus\{a\}} p^*_{new}(n_v)\right)\\
    &=\sum_{x_L^{new},n_V^{new}}\,\,\prod_{v\in V} p^*_{new}(x_v|x^{new}_{\pa_{\mathcal{G}}(v)}, n_v^{new})\prod_{l\in L} p^*_{new}(x_l^{new})\prod_{v\in V} p^*_{new}(n_v^{new}).
\end{align*}
Note that the response distribution of $a$ is still deterministic and hence real nonnegative in $p^*_{new}$. Set $p^*:=p^*_{new}$ and keep repeating the process until all noise variables have been trivialised.

At the end of this process, every noise variable in the model is trivial, so we can ignore them and obtain an assignment that achieves $p(x_V)$. Additionally, each observed variable has a response distribution that is either deterministic (if it has latent parents), or in any case real nonnegative (if it has no latent parents). Therefore, the constructed assignment witnesses $p\in\mathcal{W}_l$.\\\\

\noindent
$\bullet\,\,\,\mathcal{W}_l\subseteq\mathcal{W}_n$\\
\textit{Idea:} Simulate a quasiprobabilistic latent variable via a classical latent variable and quasiprobabilistic noise variables for its children. Given any quasiprobabilistic finite discrete variable $\Lambda$, we can split its domain $D_\Lambda$ into $D_\Lambda=D_\Lambda^+\dot{\cup}D_\Lambda^-$ as follows:
\begin{align*}
    &D_\Lambda^+=\{\lambda\in D_\Lambda:\,p(\lambda)\geq 0\} &D_\Lambda^-=\{\lambda\in D_\Lambda:\,p(\lambda)< 0\}.
\end{align*}
Now let $p^+_\Lambda=\sum_{\lambda\in D_\Lambda^+} p(\lambda)$, $p^-_\Lambda=\sum_{\lambda\in D_\Lambda^-} p(\lambda)$. These are finite quantities satisfying $p^+_\Lambda+p^-_\Lambda=1$, $p^+_\Lambda\geq 1$, $p^-_\Lambda\leq 0$, with $p^-_\Lambda=0$ iff $\Lambda$ is nonnegative.

We can construct the following variables:
\begin{itemize}
    \item[$\bullet$] $\Lambda_B$ having $D_{\Lambda_B}=\{+,-\}$, with $\mathbb{P}(\Lambda_B=+)=p^+_\Lambda$, $\mathbb{P}(\Lambda_B=-)=p^-_\Lambda$;\\
    \item[$\bullet$] $\Lambda^+$ having $D_{\Lambda^+}=D_\Lambda^+$, $\mathbb{P}(\Lambda^+=\lambda^+)=\frac{p(\lambda^+)}{p^+_\Lambda}$;\\
    \item[$\bullet$] $\Lambda^-$ having $D_{\Lambda^-}=D_\Lambda^-$, $\mathbb{P}(\Lambda^-=\lambda^-)=\frac{p(\lambda^-)}{p^-_\Lambda}$ if $p_{\Lambda}^-\neq0$, and having a trivial singleton domain (and hence trivial distribution) otherwise.
\end{itemize}

Then one can see that sampling $\Lambda$ is the same as sampling $\Lambda_B$ and then sampling $\Lambda^+$ if the result was $+$ and $\Lambda^-$ if the result was $-$.

\begin{definition}[Ensemble of Components]
    We call the random variable $(\Lambda_B,\Lambda^+,\Lambda^-)$ the \emph{ensemble of components} of quasiprobabilistic finite discrete random variable $\Lambda$.
\end{definition}

\begin{remark}
    If the variable $\Lambda$ is nonnegative, then it is isomorphic to its ensemble of components.
\end{remark}

Now consider $p\in\mathcal{W}_l$ and an assignment witnessing this. Lemma \ref{proof_split_quasilatents_assignment} shows that there exists another assignment witnessing $p\in\mathcal{W}_l$ that replaces each latent variable by its corresponding ensemble of components. We will use this latter assignment (which we denote $p^*$) to prove that $p\in\mathcal{W}_n$.

The crucial step of the proof is to show that a binary quasiprobabilistic latent can be simulated by its children if they have access to quasiprobabilistic noise and a nonnegative shared latent variable; we will see this can be made precise using Lemma \ref{proof_local_noise_simulating_shared_negativity}. We will replace each quasiprobabilistic latent variable in $p^*$ by a nonnegative one and will add noise variables to its children which will allow them to simulate the latent's original behaviour. After we have done this for each latent, $p^*$ will become a noisy assignment witnessing $p\in\mathcal{W}_n$. 

We start with the assignment $p^*$ above, which has each latent variable replaced by its ensemble of components. We can also consider this $p^*$ to have a trivial 'noise variable' for each observed vertex, i.e. each vertex $v\in V$ has a 'noise variable' $N_v$ that can take only one value $1_v$ with probability $1$. Therefore, we can define $$p^*(x_v|x_{pa_{\mathcal{G}}(v)},1_v):=p^*(x_v|x_{pa_{\mathcal{G}}(v)}).$$
It is worth mentioning that these are not proper noise variables, and we cannot yet view $p^*$ as a noisy assignment, because the response distributions defined above are not deterministic.

We also have:
\begin{equation}    \label{noise_achieving_goal}
p(x_V) = \sum_{x_L,n_V}\,\,\prod_{v\in V} p^*(x_v|x_{\pa_{\mathcal{G}}(v)}, n_v)\prod_{l\in L} p^*(x_l)\prod_{v\in V} p^*(n_v),
\end{equation}
with the reminder that the sum over each $n_v$ is trivial.

Now we will describe a process to replace a latent variable by a nonnegative one while preserving the equality \eqref{noise_achieving_goal}. Pick a latent node $l=(l_B,l^+,l^-)\in L$, with $|\ch(l)|=m.$ Assume $l$ is not nonnegative, otherwise simply move on to the next latent. Now the idea (pictured in Fig. \ref{fig:shared_binary_variable_construction}) is: sharing the variable $l_B$ between the members of $\ch(l)$ is the same as each variable $v\in\ch(l)$ having access to a binary quantity $A_v\in\{+,-\}$, where the variables $\{A_v\}_{v\in \ch(l)}$ are distributed as \begin{equation}
    p_{A_1,\dots,A_m}=p_{l_B}(+)\cdot[++\dots+] + p_{l_B}(-)\cdot [--\dots -],  \label{eq:p_As}
\end{equation} where $[++\dots+]$ is the deterministic distribution assigning $+$ to all variables $A_1,\dots,A_m$ with probability $1$. By Lemma \ref{proof_local_noise_simulating_shared_negativity}, this joint distribution over the $A_v$s can be simulated using a shared classical latent variable $\Lambda^{(l)}$ and some quasiprobabilistic noise variables $N^{(l)}_v$. However, instead of calculating the value of $A_v$ and then sending it to $v$, one might as well just send $\lambda^{(l)}$ to all children and give $v$ access to $N^{(l)}_v$; that way, each $v$ can calculate the value of $A_v$ internally.

\begin{figure}
    \centering
    \ctikzfig{TikzFigures/shared_binary_variable_construction}
    \caption{The process of replacing a quasiprobabilistic latent $l$ with quasiprobabilistic noise. The process involves replacing the binary flag $l_B$ in the ensemble of components of the latent using Lemma \ref{proof_local_noise_simulating_shared_negativity}, and then absorbing the resulting variables $A_1,\dots,A_m$ into the observed children of $l$.}
    \label{fig:shared_binary_variable_construction}
\end{figure}

Formally, apply Lemma \ref{proof_local_noise_simulating_shared_negativity} for $n:=m$ and $(p^+,p^-):=\Big(p^*(l_B=+),\,p^*(l_B=-)\Big)$, and consider $p^{(l)}$ to be a noisy assignment witnessing $p_{A_1,\dots A_m}\in \mathcal{W}_n(Com_m)$, where the graph $Com_m$ can be found in Fig. \ref{fig:lemma_noise_simulating_negativity} and $p_{A_1,\dots, A_m}$ is the distribution in eq \eqref{eq:p_As} over some binary variables $A_1,\dots,A_m$. Note that the lemma uses $\mathcal{W}_r$ instead, but we have already shown that $\mathcal{W}_r=\mathcal{W}_n$.

Note that assignment $p^{(l)}$ has only one latent variable, which we will call $\Lambda^{(l)}$. We will also assign one of the variables $A_1,\dots, A_m$ to each child of $l$; the one assigned to $v\in\ch(l)$ will be called $A_v$. The noise variable associated with $A_v$ in $p^{(l)}$ will be called $N_v^{(l)}$. The deterministic value of $A_v$ given $\lambda^{(l)}$ and $n_v^{(l)}$ is denoted $f_v(\lambda^{(l)},n_v^{(l)})$.

Now construct a new assignment $p^*_{new}$ by replacing in $p^*$ the variable $l_B$ by variable $\Lambda^{(l)}$, and for each $v\in\ch(l)$ replace $N_v$ by $(N_v,N_v^{(l)}).$ Also, set for each $v\in \ch(l)$
\begin{gather*}
p^*_{new}\left(x_v\bigg|x_{\pa_{\mathcal{G}}(v)\setminus \{l\}}, (\lambda^{(l)},x_l^+,x_l^-), (n_v,n_v^{(l)})\right)\\
\hspace{2cm}:=p^*\left(x_v|x_{\pa_{\mathcal{G}}(v)\setminus \{l\}},\big(f_v(\lambda^{(l)},n_v^{(l)}),x^+,x^-\big),n_v\right),
\end{gather*}
and $p^*_{new}(\lambda^{(l)}):=p^{(l)}(\lambda_l)$. Otherwise have $p^*_{new}$ behave the same as $p^*$. We will show $p^*_{new}$ achieves $p(x_V)$ as well.

Note that for all $x_V,x_l^+,x_l^-,n_V$, the following holds:
\begin{equation}\label{substitution_equation_lemma}
\begin{gathered}
\sum_{x_{l_B}}p^*(x_{l_B})\prod_{v\in\ch(l)}p^*(x_v|x_{\pa_{\mathcal{G}}(v)\setminus \{l\}},(x_{l_B},x_l^+,x_l^-),n_v)\\
=\sum_{a_1,\dots,a_m}p_{A_1,\dots,A_m}(a_1,\dots,a_m)\prod_{v\in\ch(l)}p^*(x_v|x_{\pa_{\mathcal{G}}(v)\setminus \{l\}},(a_v,x_l^+,x_l^-),n_v)\\
=\sum_{a_1,\dots,a_m}\bigg(\sum_{\lambda^{(l)},n_{\ch(l)}^{(l)}}p^{(l)}(\lambda^{(l)})\prod_{v\in\ch(l)}p^{(l)}(a_v|\lambda^{(l)},n_v^{(l)})p^{(l)}(n_v^{(l)})\bigg)\prod_{v\in\ch(l)}p^*(x_v|x_{\pa_{\mathcal{G}}(v)\setminus \{l\}},(a_v,x_l^+,x_l^-),n_v)\\
=\sum_{\lambda^{(l)},n_{\ch(l)}^{(l)}}p^{(l)}(\lambda^{(l)})\prod_{v\in\ch(l)}p^{(l)}(n_v^{(l)})\prod_{v\in\ch(l)}p^*(x_v|x_{\pa_{\mathcal{G}}(v)\setminus \{l\}},\big(f_v(\lambda^{(l)},n_v^{(l)}),x_l^+,x_l^-\big),n_v)\\
=\sum_{\lambda^{(l)},n_{\ch(l)}^{(l)}}p^{(l)}(\lambda^{(l)})\prod_{v\in\ch(l)}p^{(l)}(n_v^{(l)})\prod_{v\in\ch(l)}p^*_{new}\left(x_v\bigg|x_{\pa_{\mathcal{G}}(v)\setminus \{l\}}, (\lambda^{(l)},x_l^+,x_l^-), (n_v,n_v^{(l)})\right),
\end{gathered}
\end{equation}
where the first equality comes from the fact that $p_{A_1,\dots,A_m}(a_1,\dots,a_m)$ is almost always $0$ (and is nonzero exactly when $a_1=\dots=a_m\in\{-,+\}$), the second comes from the fact that $p_{A_1,\dots A_m}$ factorises with respect to $Com$, and the third by using the fact that $p^{(l)}(a_v|\lambda^{(l)},n_v^{(l)})$ is a deterministic distribution. Using this, we can prove that the assignment $p^*_{new}$ also achieves $p$:
\begin{align*}
p(x_V)
&=\sum_{x_L,n_V}\,\,\prod_{v\in V} p^*(x_v|x_{\pa_{\mathcal{G}}(v)}, n_v)\prod_{l'\in L} p^*(x_{l'})\prod_{v\in V} p^*(n_v)\\
&=\sum_{x_{L\setminus\{l\}},n_{V\setminus\ch(l)}}\,\,\prod_{v\in V\setminus ch(l)}p^*(x_v|x_{\pa_{\mathcal{G}}(v)}, n_v)\prod_{l'\in L\setminus\{l\}} p^*(x_{l'})\prod_{v\in V\setminus\ch(l)} p^*(n_v)\\
& \sum_{x_{l_B},x_l^+,x_l^-,n_{\ch(l)}}p^*(x_{l_B})p^*(x_l^+)p^*(x_l^-)\prod_{v\in\ch(l)}p^*\left(x_v|x_{\pa_{\mathcal{G}}(v)\setminus \{l\}},(x_{l_B},x_l^+,x_l^-),n_v\right)\cdot p^*(n_v)\\
&\overset{\eqref{substitution_equation_lemma}}{=}\sum_{x_{L\setminus\{l\}},n_{V\setminus\ch(l)}}\,\,\prod_{v\in V\setminus ch(l)}p^*(x_v|x_{\pa_{\mathcal{G}}(v)}, n_v)\prod_{l'\in L\setminus\{l\}} p^*(x_{l'})\prod_{v\in V\setminus\ch(l)} p^*(n_v)\\
& \sum_{\lambda^{(l)},x_l^+,x_l^-,n_{\ch(l)},n_{\ch(l)}^{(l)}}p^{(l)}(\lambda^{(l)})p^*(x_l^+)p^*(x_l^-)\\
& \prod_{v\in\ch(l)}p^*_{new}\left(x_v\bigg|x_{\pa_{\mathcal{G}}(v)\setminus \{l\}}, (\lambda^{(l)},x_l^+,x_l^-), (n_v,n_v^{(l)})\right)p^*(n_v)p^{(l)}(n_v^{(l)})\\
&=\sum_{x_L^{new},n_V^{new}}\,\,\prod_{v\in V} p^*_{new}(x_v|x_{\pa_{\mathcal{G}}(v)}, n_v^{new})\prod_{l\in L} p^*_{new}(x_l^{new})\prod_{v\in V} p^*_{new}(n_v^{new}).
\end{align*}

Hence, $p^*_{new}$ also achieves $p$, but has a nonnegative distribution for $l$. Replace $p^*$ by $p^*_{new}$ and repeat the process for all latents. At the end, we have an almost-noisy-assignment achieving $p$. The only issue is the response functions are not deterministic, but this can be fixed by extending the noise variables to also contain the response distributions of the observed variables (following the construction we used to prove $\mathcal{W}_r\subseteq\mathcal{W}_n$). Finally, this gives a noisy assignment that achieves $p$, and hence witnesses $p\in\mathcal{W}_n$. 

Therefore, $\mathcal{W}_l=\mathcal{W}_n$ and our proof is done.
\end{proof}

\begin{lemma} \label{proof_no_latent_parents}
Consider a DAG $\mathcal{G}$ with vertices $V\dot{\cup}L$, and $p$ a distribution that can be written in the form 
\begin{equation}
p(x_V)=\sum_{x_L}\prod_{v\in V} p_1(x_v|x_{\pa_{\mathcal{G}}(v)})\prod_{l\in L}p_1(x_l),\label{lemma_factorisation_expression}
\end{equation}
for quasidistributions $p_1(x_v|x_{\pa_\mathcal{G}(v)})$ and $p_1(x_l)$. Then for any vertex $a\in V$ that has no latent parents, $$p_1(x_a|x_{\pa_\mathcal{G}(a)})=p(x_a|x_{\pa_\mathcal{G}(a)}).$$
\end{lemma}
\begin{proof}
We can assume WLOG that $a$ has no children. Otherwise, simply sum out all its descendants (except $a$) in \eqref{lemma_factorisation_expression} to obtain that the remaining distribution factorises with respect to the new (ancestral) graph.

Now, denote by $A:=V\setminus(\pa_\mathcal{G}(a)\cup\{a\})$. Then,
\begin{align*}
p(x_a,x_{\pa_{\mathcal{G}}(a)})
    &=\sum_{x_L,x_A} \prod_{v\in V} p_1(x_v|x_{\pa_{\mathcal{G}}(v)})\prod_{l\in L}p_1(x_l)\\
    &=p_1(x_a|x_{\pa_\mathcal{G}(a)})\sum_{x_L,x_A} \prod_{\substack{v\in V\\v\neq a}} p_1(x_v|x_{\pa_{\mathcal{G}}(v)})\prod_{l\in L}p_1(x_l).
\end{align*}

As a result,
\begin{align*}
    p(x_{\pa_{\mathcal{G}}(a)})&=\sum_{x_a}p(x_a,x_{\pa_{\mathcal{G}}(a)})\\
    &=\sum_{x_L,x_A} \prod_{\substack{v\in V\\v\neq a}} p_1(x_v|x_{\pa_{\mathcal{G}}(v)})\prod_{l\in L}p_1(x_l).
\end{align*}

Dividing the above two expressions, we obtain $$
    p(x_a|x_{\pa_\mathcal{G}(a)})
    = \frac{p(x_a,x_{\pa_{\mathcal{G}}(a)})}{p(x_{\pa_{\mathcal{G}}(a)})}
    = p_1(x_a|x_{\pa_\mathcal{G}(a)}).$$
\end{proof}

\begin{lemma} \label{proof_split_quasilatents_assignment}
Consider a DAG $\mathcal{G}$ and an assignment $p^*$ witnessing $p\in\mathcal{W}_l(\mathcal{G})$ for a distribution $p(x_V)$ over $\mathcal{G}$'s vertices. Then there exists assignment $p^{**}$ that also witnesses $p\in\mathcal{W}_l(\mathcal{G})$, which replaces every latent in $p^*$ by its ensemble of components. 
\end{lemma}
\begin{proof}
We will iteratively replace each latent by its ensemble of components. Pick a latent node $l\in L$. Assume $l$ is not nonnegative; otherwise it is isomorphic to its own ensemble of components, so we are done. Consider a new assignment $p^*_{new}$ which has the variable $X_l$ replaced by its corresponding ensemble of components $(X_{l_B},X_l^+,X_l^-)$, and set $$p^*_{new}(x_{l_B},x_l^+,x_l^-):=p^{x_{l_B}}_l\cdot\frac{p^*(x_l^+)}{p_l^+}\cdot\frac{p^*(x_l^-)}{p_l^-},$$ where $p^{x_{l_B}}_l:=
\begin{cases}
    p_l^+,\text{ if }x_{l_B}=+\\
    p_l^-,\text{ if }x_{l_B}=-
\end{cases},$ $p_l^+=\sum_{x_l;\,p^*(x_l)\geq 0} p^*(x_l)$, and
$p_l^-=\sum_{x_l;\,p^*(x_l)< 0} p^*(x_l)$.

For each observed child $v\in \ch(l)$ of $l$, set $$p^*_{new}(x_v|x_{\pa_{\mathcal{G}}(v)\setminus\{l\}},x_{l_B},x_l^+,x_l^-):=p^*\left(x_v|x_{\pa_{\mathcal{G}}(v)\setminus\{l\}},x_l=x^{x_{l_B}}_l\right),$$
where by $x_l=x_l^{x_{l_B}}$ in the RHS we mean $x_l=x_l^+$ if $x_{l_B}=+$, and $x_l=x_l^-$ if $x_{l_B}=-$.

Set everything else in $p^*_{new}$ to behave the same as in $p^*$. Then $p^*_{new}$ also achieves $p(x_V)$:
\begin{align*}
p(x_V)
    &=\sum_{x_L}\prod_{v\in V} p^*(x_v|x_{\pa_{\mathcal{G}}(v)})\prod_{l'\in L} p^*(x_{l'})\\
    &=\sum_{x_{L\setminus\{l\}}}\prod_{l'\in L\setminus\{l\}}p^*(x_{l'})\prod_{v\in V\setminus\ch(l)}p^*(x_v|x_{\pa_{\mathcal{G}}(v)})\sum_{x_l}p^*(x_l)\prod_{v\in\ch(l)}p^*(x_v|x_{\pa_{\mathcal{G}}(v)})\\
    &=\sum_{x_{L\setminus\{l\}}}\prod_{l'\in L\setminus\{l\}}p^*_{new}(x_{l'})\prod_{v\in V\setminus\ch(l)}p^*_{new}(x_v|x_{\pa_{\mathcal{G}}(v)})\\
    &\left(\sum_{x_l^+} p^*(x_l^+)\prod_{v\in\ch(l)}p^*(x_v|x_{\pa_{\mathcal{G}}(v)\setminus\{l\}},x_l^+)+\sum_{x_l^-} p^*(x_l^-)\prod_{v\in\ch(l)}p^*(x_v|x_{\pa_{\mathcal{G}}(v)\setminus\{l\}},x_l^-)\right)\\
    &=\sum_{x_{L\setminus\{l\}}}\prod_{l'\in L\setminus\{l\}}p^*_{new}(x_{l'})\prod_{v\in V\setminus\ch(l)}p^*_{new}(x_v|x_{\pa_{\mathcal{G}}(v)})\\
    &\Bigg(p_l^+\sum_{x_l^+} \frac{p^*(x_l^+)}{p_l^+}\prod_{v\in\ch(l)}p^*(x_v|x_{\pa_{\mathcal{G}}(v)\setminus\{l\}},x_l^+)\cdot\sum_{x_l^-} \frac{p^*(x_l^-)}{p_l^-}\\
    &+
    p_l^-\sum_{x_l^-} \frac{p^*(x_l^-)}{p_l^-}\prod_{v\in\ch(l)}p^*(x_v|x_{\pa_{\mathcal{G}}(v)\setminus\{l\}},x_l^-)\cdot\sum_{x_l^+} \frac{p^*(x_l^+)}{p_l^+}\Bigg)\\
    &=\sum_{x_{L\setminus\{l\}}}\prod_{l'\in L\setminus\{l\}}p^*_{new}(x_{l'})\prod_{V\setminus\ch(l)}p^*_{new}(x_v|x_{\pa_{\mathcal{G}}(v)})\\
    &\sum_{x_{l_B},x_l^+,x_l^-} p^*_{new}(x_{l_B},x_l^+,x_l^-)\prod_{v\in\ch(l)}p^*_{new}(x_v|x_{\pa_{\mathcal{G}}(v)\setminus\{l\}},x_{l_B},x_l^+,x_l^-)\\
    &=\sum_{x_L^{new}}\prod_{v\in V} p^*_{new}(x_v|x_{\pa_{\mathcal{G}}(v)}^{new})\prod_{l'\in L} p^*_{new}(x_{l'}^{new}),
\end{align*}
and the response distributions of the observed variables are still real nonnegative, therefore $p^*_{new}$ witnesses $p\in\mathcal{W}_l$. Replace $p^*$ by $p^*_{new}$ and repeat the process for the next latent variable, until none are left. We define $p^{**}$ to be the final assignment we obtain, and we are done.
\end{proof}

\begin{figure}[t]
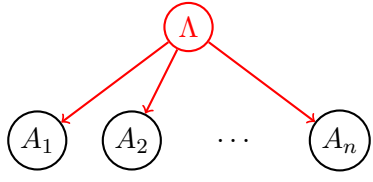
  
    \ctikzfig{TikzFigures/Lemma_local_noise_simulating_shared_negativity}
    \caption{Graph $Com_n$.}
    \label{fig:lemma_noise_simulating_negativity}
\end{figure}

\begin{lemma}   \label{proof_local_noise_simulating_shared_negativity}
Consider binary random variables $A_1,\dots,A_n$ taking values in $\{+,-\}$. We denote $[a_1,\dots,a_n]$ to mean the deterministic distribution assigning values $a_1,\dots,a_n$ to the variables. Then the distribution $$p_{A_1,\dots,A_n}=p^+\cdot[++\dots+]\,+\, p^-\cdot[--\dots-]$$ is in the set\footnote{Although this is technically slightly imprecise, since we assumed the sets $\mathcal{W}_l,\mathcal{W}_r,\mathcal{W}_n$ only contain nonnegative observed distributions, but we ignore that detail here.} $\mathcal{W}_r$ of the graph $Com_n$ (which stands for 'common ancestor') in Figure \ref{fig:lemma_noise_simulating_negativity} for any $p^+,p^-\in\mathbb{R}$ such that $p^++p^-=1$, $p^+\geq 1$, $p^-\leq0$.
\end{lemma}
\begin{proof}
We will prove via induction over $n$ that there exists an assignment $p^*$ witnessing $p_{A_1,\dots,A_n}\in\mathcal{W}_r$ that uses a latent variable $\Lambda$ that is uniform over $2^{n-1}$ possible outcomes.

Let $D_\Lambda=\{1,\dots,2^{n-1}\}$ and $$p^*(\lambda)=\frac{1}{2^{n-1}},\,\,\forall \lambda\in D_\Lambda.$$
Denote also
\begin{align*}
    &p^*(a_i=+\,|\,\lambda)=:x_i^{(\lambda)} & p^*(a_i=-\,|\,\lambda)=:y_i^{(\lambda)},
\end{align*}
where $x_i^{(\lambda)}+y_i^{(\lambda)}=1$ for all $i$ and $\lambda$.

Therefore, the condition $$p_{A_1,\dots,A_n}(a_1,\dots,a_n)=\sum_{\lambda\in D_{\Lambda}}p^*(\lambda)\prod_{i=1}^n p^*(a_i|\lambda)$$
is equivalent to the following system:
\begin{equation}    \label{system_lemma}
\begin{cases}
2^{n-1}p^+&=\sum_{\lambda}\prod_{i=1}^n x_i^{(\lambda)}\\\\
0&=\sum_\lambda\prod_{i\in S^+} x_i^{(\lambda)}\prod_{i\in S^-} y_i^{(\lambda)}
\quad,\forall\, S^+\dot\cup S^-=\{1,\dots,n\}, S^+\neq0, S^-\neq 0\\\\
2^{n-1}-2^{n-1}p^+&=\sum_{\lambda}\prod_{i=1}^n y_i^{(\lambda)}.
\end{cases}
\end{equation}
The first equation comes from $p_{A_1,\dots,A_N}(+,\dots,+)=p^+$, the last one comes from $p_{A_1,\dots,A_n}(-,\dots,-)=p^-$, and the ones in the middle make sure that $p_{A_1,\dots,A_n}(a_1,\dots,a_n)=0$ whenever the $a_i$s are not all equal. Note the important condition $S^+\neq0$ and $S^-\neq0$ in the second row. This ensures these equations do not overlap with the $++\dots+$ and $--\dots-$ cases.

It is therefore sufficient to prove that system (\ref{system_lemma}) has solutions $x_i^{(\lambda)},y_i^{(\lambda)}$ over reals. We can slightly alter the system to bring it to a neater form.

To see this, if we consider $S^-:=\{j\}$, the corresponding equation is 
\begin{gather*}
    \sum_\lambda\left(\prod_{i\neq j}^n x_i^{(\lambda)}\right)y_j^{(\lambda)}=0\\
    \iff \sum_\lambda\left(\prod_{i\neq j}^n x_i^{(\lambda)}\right)\left(1-x_j^{(\lambda)}\right)=0\\
    \iff \sum_\lambda\left(\prod_{i\neq j}^n x_i^{(\lambda)}\right)-2^{n-1}p^+=0\\
    \iff \sum_\lambda\prod_{i\in S^+} x_i^{(\lambda)}=2^{n-1}p^+,
\end{gather*}
where the second-to-last equivalence comes from using the first equation of system \eqref{system_lemma}. Therefore, because we are dealing with a linear system of equations, it is sound to replace the equations that have $|S^+|=n-1$ with $$\sum_\lambda\prod_{i\in S^+} x_i^{(\lambda)}=2^{n-1}p^+.$$

However, we can now proceed inductively to prove that these equations being true for all $S^+\neq\emptyset$ is enough to conclude we have solved the initial system. To see this, assume we proved it for all $|S^+|\geq k>1$. Pick $S^+$ such that $|S^+|=k-1$. Then the corresponding equation is:

\begin{gather*}
\sum_\lambda\prod_{i\in S^+}x_i^{(\lambda)}\prod_{i\in S^-}y_i^{(\lambda)}=0\\
\iff \sum_\lambda \prod_{i\in S^+} x_i^{(\lambda)}\prod_{i\in S^-}\left(1-x_i^{(\lambda)}\right)=0\\
\iff \sum_\lambda\left(\prod_{i\in S^+} x_i^{(\lambda)} - \sum_{j\in S^-}\prod_{i\in S^+\cup\{j\}} x_i^{(\lambda)}+\sum_{j,k\in S^-}\prod_{i\in S^+\cup\{j,k\}} x_i^{(\lambda)}-\dots\right)=0
\end{gather*}
But by induction hypothesis, all of the products in the bracket except the first are equal to $2^{n-1}p^+$. By the fact that the number of even subsets of a non-empty set is equal to the number of odd subsets, all but one of the products end up cancelling each other, leaving as desired the equation $$\sum_\lambda\prod_{i\in S^+}x_i^{(\lambda)}=2^{n-1}p^+.$$

Therefore, we can replace our system by the equivalent system:
\begin{equation*}
\begin{cases}
\sum_{\lambda}\prod_{i\in S^+} x_i^{(\lambda)}=2^{n-1}p^+
\hspace{1cm},\forall\, S^+\subseteq\{1,\dots,n\}, S^+\neq0\\\\
\sum_{\lambda}\prod_{i=1}^n \left(1-x_i^{(\lambda)}\right)=2^{n-1}-2^{n-1}p^+.
\end{cases}
\end{equation*}
In fact, we can completely drop the last equation, since it is a consequence of the other equations and the fact that our conditional distributions are normalised. Therefore, the actual system we will solve over reals $x_i^{(\lambda)}$ is 
\begin{equation}    \label{lemma_final_system}
\sum_{\lambda}\prod_{i\in S} x_i^{(\lambda)}=2^{n-1}p^+
\hspace{1cm},\forall\, S\subseteq\{1,\dots,n\}, S\neq0.
\end{equation}

As mentioned before, we will prove this system has solutions over $\mathbb{R}$ for any $p^+\geq1$ using induction over $n$. For the base case, if $n=1$ the system is simply $x_1^{(1)}=p^+$.

For the induction step, assume the system has solutions for any $p^+\geq 1$ for $n-1$ variables and we will prove this is the case for $n$ as well. Now pick $\{x_i^{(\lambda)}\}_{i=\overline{1..n-1},\,\lambda=\overline{1..2^{n-2}}}$ to be a solution for the system with $n-1$ variables with some value $p^+_1\geq 1$ and similarly $\{x_i^{(\lambda)}\}_{i=\overline{1..n-1},\,\lambda=\overline{(2^{n-2}+1)..2^{n-1}}}$ a solution for the system with $n-1$ variables and another value $p^+_2\geq1$; we will pick values for $p^+_1$ and $p^+_2$ later.

\begin{table}[t]
\centering
\begin{tabular}{|c| >{\centering\arraybackslash}p{0.75cm} >{\centering\arraybackslash}p{0.75cm} >{\centering\arraybackslash}p{0.75cm} >{\centering\arraybackslash}p{0.75cm} | >{\centering\arraybackslash}p{1.5cm} >{\centering\arraybackslash}p{0.75cm} >{\centering\arraybackslash}p{0.75cm} |}
    \hline
    & \multicolumn{7}{c|}{$\lambda$}\\
    \cline{2-8}
    & $1$ & $2$ & \dots & $2^{n-2}$  &   $2^{n-2}+1$ & \dots & $2^{n-1}$\\
    \hline
    $x_n^{(\lambda)}$ & $\alpha$ & $\alpha$ & \dots & $\alpha$ & $\beta$ & \dots & $\beta$\\
    \cline{2-8}
    $x_{n-1}^{(\lambda)}$ & & & & & & & \\
    \vdots & \multicolumn{4}{>{\centering\arraybackslash}>{\centering\arraybackslash}p{4cm}|}{solution for $n-1$ variables and $p^+_1$} & \multicolumn{3}{>{\centering\arraybackslash}p{4cm}|}{solution for $n-1$ variables and $p^+_2$}\\
    $x_1^{(\lambda)}$ & & & & & & & \\
    \hline
\end{tabular}
\vspace{0.2cm}
\caption{A table showing the inductive construction of the system solution.}
\label{solution_table}
\end{table}

We also pick $x_n^{(\lambda)}=\begin{cases}
    \alpha, \text{ if } \lambda\leq 2^{n-2}\\
    \beta, \text{ otherwise}
\end{cases}$. Once again, we leave $\alpha$ and $\beta$ unfixed for now and will choose convenient values later. See Table \ref{solution_table} for a visualisation of the construction.

Now for any $\emptyset\neq S\subseteq\{1,\dots,n-1\}$,
\begin{align*}
\sum_{\lambda=1}^{2^{n-1}}\prod_{i\in S} x_i^{(\lambda)}
    &=\sum_{\lambda=1}^{2^{n-2}}\prod_{i\in S} x_i^{(\lambda)} + \sum_{\lambda=2^{n-2}+1}^{2^{n-1}}\,\,\prod_{i\in S} x_i^{(\lambda)}\\
    &=2^{n-2}p^+_1 + 2^{n-2}p^+_2.
\end{align*}
The last equality holds because of the induction hypothesis and our construction. Therefore, when $\emptyset\neq S\subseteq\{1,\dots,n-1\}$, $$\sum_{\lambda}^{2^{n-1}}\prod_{i\in S} x_i^{(\lambda)}=2^{n-1}p^+\iff \boxed{p^+_1+p^+_2=2p^+}.$$

If $S$ contains $n$, then we need to consider two cases. If $S\setminus\{n\}\neq \emptyset$, we get:
\begin{align*}
\sum_{\lambda=1}^{2^{n-1}}\prod_{i\in S}x_i^{(\lambda)}
    &=\sum_{\lambda=1}^{2^{n-1}} x_n^{(\lambda)} \prod_{i\in S\setminus\{n\}}x_i^{(\lambda)}\\
    &=\sum_{\lambda=1}^{2^{n-2}} \alpha \cdot\prod_{i\in S\setminus\{n\}}x_i^{(\lambda)}+\sum_{\lambda=2^{n-2}+1}^{2^{n-1}} \beta \cdot\prod_{i\in S\setminus\{n\}}x_i^{(\lambda)}\\
    &=\alpha\cdot2^{n-2}p^+_1+\beta\cdot2^{n-2}p^+_2
\end{align*}
The last equality holds because of the induction hypothesis and the fact that $S\setminus\{n\}\neq\emptyset$. Therefore, when $\{n\}\subsetneq S$, $$\sum_{\lambda}^{2^{n-1}}\prod_{i\in S} x_i^{(\lambda)}=2^{n-1}p^+\iff \boxed{\alpha p^+_1+\beta p^+_2=2p^+}.$$

If $S=\{n\}$, we get:
    $$\sum_{\lambda=1}^{2^{n-1}} x_n^{(\lambda)}=2^{n-2}\alpha+2^{n-2}\beta=2^{n-1}p^+\iff \boxed{\alpha+\beta=2p^+}.$$

Therefore, if we find $p^+_1,p^+_2\geq 1$ and $\alpha,\beta\in\mathbb{R}$ that satisfy the system 
\begin{equation*}
\begin{cases}
p^+_1+p^+_2=2p^+\\
\alpha p^+_1+\beta p^+_2=2p^+\\
\alpha+\beta=2p^+
\end{cases},
\end{equation*}
we are done. This is actually the system for $n=2$, and indeed has solution $p^+_1=1,p^+_2=2p^+-1\geq1$ and $\alpha=2p^+,\beta=0$.

Building the solution as described above therefore attests that the system has solutions for $n$ as well and our induction is complete. Therefore, the system has solutions for any natural number $n$, which means we can always solve it to construct an assignment that achieves $p_{A_1,\dots,A_n}$. Note that if $\alpha=2p^+>1$, then this will mean that $y_n^{(\lambda)}$ will be negative for some values of $\lambda$, which will translate to some negative response distributions. However, we set the variable $\Lambda$ to be uniform, and therefore positive, so the assignment we obtain at the end will indeed witness $p_{A_1,\dots,A_N}\in\mathcal{W}_r$.
\end{proof}

\section{Proof of Proposition \ref{proposition:quasi_is_in_nested}:} \label{sec:proof_quasi_is_in_nested}
In this section we will prove that $\mathcal{W}(\mathcal{G})\subseteq\mathcal{N}(\mathcal{G})$ by following some arguments in \cite{nested}. We will not define all the notions used in the proof and direct the reader to the cited paper for more details.

To start off, one can generalise the notion of a kernel to the notion of \emph{quasikernel} in the natural way. We will assume soundness of semigraphoid axioms for quasikernels, given that the proof of this in \cite{nested} uses only symbolic manipulations. Similarly, one can extend the fixing operation $\phi_\mathbf{w}(\_;\mathcal{G})$ to be defined on quasikernels by simply following the same formula as the one used for kernels.

Classically, the nested Markov model of a CADMG $\mathcal{G}$ with observed vertices $V$ and fixed vertices $W$ is defined as the set of kernels $q_{V|W}(x_V|x_W)$ with the property that for any fixing sequence $\mathbf{w}$, the kernel $\phi_\mathbf{w}(q;\mathcal{G})$ satisfies the $4$ equivalent Markov properties with respect to the fixed graph $\phi_\mathbf{w}(\mathcal{G})$, including the local Markov property. The local Markov property enforces some equality conditions on the fixed kernel, which is a rational function of the initial kernel, so all equalities enforced by the nested Markov model can be seen as (rational) constraints on $q$.

Now note that the proof of Theorem B.5 in \cite{nested} uses only factorisation and the chain rule of probabilities, so it is sound for quasikernels as well, meaning that Tian factorisation implies the local Markov property for any fixing operation $\mathbf{w}$ and quasi-kernel $\phi_{\textbf{w}}(q)$.

Finally, one can see that all quasikernels $\phi_{\textbf{w}}(q)$ that can be obtained from some $q\in\mathcal{W}(\mathcal{G})$ by applying a fixing sequence must Tian factorise (which can  be seen by grouping terms in the original factorisation), meaning they satisfy the local Markov property. Translating these into constraints on the initial quasidistribution, we get that all nested Markov constraints are satisfied.

\section{Conjecture \ref{conjecture} for Tree-Structured Correlation Scenarios}
    \label{sec:proof_conjecture_trees}

\subsection{Tensor Multilinear Products}    \label{sec:multilinear_products}
In this section of the appendix, we include an introduction to the tensor operations we will be using.

Let $k$ be a field of characteristic $0$. A $n_1\times\dots\times n_q$ tensor $\mathcal{T}$ with entries in $k$ is called \emph{simple} or \emph{elementary} or \emph{rank-1} if it is the tensor product of a bunch of vectors, i.e. there exist vectors $v_1,\dots, v_q$ such that $$\mathcal{T}=v_1\otimes v_2\otimes\dots\otimes v_q.$$

Let $\mathcal{T}$ be an arbitrary $n_1\times\dots\times n_q$ tensor with entries in $k$. Then $\mathcal{T}$ has a \emph{tensor rank}, which is a minimal positive integer $r$ such that $$\mathcal{T}=\sum_{i=1}^r v^{(i)}_1\otimes v^{(i)}_2\otimes\dots\otimes v^{(i)}_q,$$ for some $k$-valued vectors $\{v^{(i)}_j\}_{1\leq i\leq r,1\leq j\leq q}$, i.e. $r$ is the minimal number of simple tensors that sum to $\mathcal{T}$. 

Note the correspondence between the rank of a tensor and the rank of a matrix. Indeed, the rank of a matrix is a particular case of a tensor rank, where the tensor has order $2$. In the case of tensors however, the rank can't be computed from the dimension of the fibre\footnote{The fibre is the generalisation of the concepts of rows and columns of a matrix} vector spaces in the same way as in the case of a matrix. In fact, for tensors these dimensions need not be equal regardless of the type of fibre we choose, whereas it is well-known the row rank and column rank are equal for every matrix. Computing the rank over the reals of an arbitrary tensor even as small as order $3$ is in fact NP-hard \cite{TensorNPHard}.

Given an arbitrary tensor $\mathcal{T}$ as before and its tensor rank representation, we can apply a linear transformation $A:k^{n_i}\rightarrow k^{m_i}$ to its $i$th component. This map sends $v_i\in k^{n_i}$ to $Av_i\in k^{m_i}$, and $v_1\otimes\dots\otimes v_i\otimes\dots\otimes v_q$ to $v_1\otimes\dots\otimes \left(A v_i\right)\otimes\dots\otimes v_q$, and so by linearity we can define its application on the arbitrary tensor $\mathcal{T}$, denoted using the symbol " $\times_{i}$", as:
$$A\times_{i}\mathcal{T} = \sum_{j=1}^r v^{(j)}_1\otimes\dots\otimes \left(Av^{(j)}_i\right)\otimes\dots\otimes v^{(j)}_q.$$
This holds even if the representation is not a tensor rank representation, i.e. even if the value $r$ is not optimal.

It is not too difficult to see that two such linear transformations acting on different components $A_i:k^{n_i}\rightarrow k^{m_i}$ and $A_j:k^{n_j}\rightarrow k^{m_j}$ commute, i.e. $$A_i\times_{i}(A_j\times_{j}\mathcal{T})=A_j\times_{j}(A_i\times_{i}\mathcal{T}).$$ In fact, given $q$ of these transformations, $\{A_i:k^{n_i}\rightarrow k^{m_i}\}_{1\leq i\leq q}$, we can define the multilinear multiplication operator $\mathcal{A}=(A_1,\dots,A_q):k^{n_1}\otimes\dots\otimes k^{n_q}\rightarrow k^{m_1}\otimes\dots\otimes k^{m_q}$ as $$\mathcal{A}\times \mathcal{T}=\sum_{i=1}^r \left(A_1v^{(i)}_1\right)\otimes \left(A_2v^{(i)}_2\right)\otimes\dots\otimes \left(A_qv^{(i)}_q\right).$$

Additionally, it's easy to check that if $A_i$ has a left-inverse $A_i^L$ for all $i$, then there is a left-inverse multilinear multiplication operator $\mathcal{A}^L=(A_1^L,\dots,A^L_q)$ such that $$\mathcal{A}^L\times\mathcal{A}\times\mathcal{T}=\mathcal{T}.$$
The right-inverse and the inverse multilinear multiplication operators are defined similarly in the way one would expect.

\subsection{Gauge Freedom of Tensor Network Decompositions}
\label{sec:gauge_transformations}
\begin{definition}[Gauge Transformations]
Let $\mathcal{G}$ be a tensor network with spaces $\{\mathcal{X}_v\}_{v\in V}$ for its nodes and bond dimensions $\{r_e\}_{e\in E}$, and $\mathcal{T}\in \bigotimes_{v\in V} k^{|\mathcal{X}_v|}$ a tensor that decomposes with respect to this construction. Let such a decomposition be given by tensors $\left\{T^{(v)}\in k^{|\mathcal{X}_v|}\otimes\left(\bigotimes_{e\in\inc(v)} k^{r_e}\right)\right\}_{v\in V}$. Now pick an edge $e^*=(a,b)\in E$ and an $r_{e^*}\times r_{e^*}$ invertible matrix $X$. Then a \emph{gauge transformation} of the decomposition consists of replacing $T^{(a)}$ by $X\times_{i_{e^*}}T^{(a)}$ and $T^{(b)}$ by $\left(X^T\right)^{-1}\times_{i_{e^*}}T^{(b)}$.
\end{definition}

\begin{proposition}[Gauge Freedom]
Consider the tensor network $\mathcal{G}$ with bond dimensions $\{r_e\}_{e\in E}$ as before, and a decomposition of $\mathcal{T}$ with respect to $\mathcal{G}$ and $\{r_e\}_{e\in E}$. Then the resulting set of tensors $\{T^{(v)}\}_{v\in V}$ after applying a gauge transformation is also a tensor decomposition of $\mathcal{T}$.
\end{proposition}
\begin{proof}
Let \begin{gather}
    T^{(a)}=\sum_{j_a=1}^{c_a}\mathbf{u}^{(j_a)}\otimes\left(\bigotimes_{e\in\inc(a)} \mathbf u^{(j_a)}_e\right),    \label{eq:tensor_decomposition_A}\\
    T^{(b)}=\sum_{j_b=1}^{c_b}\mathbf w^{(j_b)}\otimes\left(\bigotimes_{e\in\inc(b)} \mathbf w^{(j_b)}_e\right) \label{eq:tensor_decomposition_B},
\end{gather}
be rank decompositions of $T^{(a)}$ and $T^{(b)}$. The idea is to show that the RHS of the tensor network decomposition formula \eqref{eq:tensor_factorisation}, namely $$\sum_{\{1\leq i_e\leq r_e\}_{e\in E}}\prod_{v\in V} T^{(v)}_{\{i_e\}_{e\in \inc(v)}}(x_v),$$ only depends on the $\mathbf{u}_{e^*}$ and $\mathbf{w}_{e^*}$ vectors via the dot products $(\mathbf{u}_{e^*})^T\cdot \mathbf{w}_{e^*}$, which remain unchanged after the gauge transformation.

More precisely, if we write the explicit forms of equations \eqref{eq:tensor_decomposition_A} and \eqref{eq:tensor_decomposition_B} we obtain: 
\begin{gather*}
    T_{\{i_e\}_{e\in\inc(a)}}^{(a)}(x_a)=\sum_{j_a=1}^{c_a}u^{(j_a)}(x_a)\cdot\left(\prod_{e\in\inc(a)} u^{(j_a)}_e(i_e)\right),\\
    T_{\{i_e\}_{e\in\inc(b)}}^{(b)}(x_b)=\sum_{j_b=1}^{c_b}w^{(j_b)}(x_b)\cdot\left(\prod_{e\in\inc(b)} w^{(j_b)}_e(i_e)\right).
\end{gather*}

Now we have
\begin{align*}
    &\sum_{i_{e^*}}T^{(a)}_{\{i_{e}\}_{e\in \inc(a)}}(x_a)\cdot T^{(b)}_{\{i_e\}_{e\in\inc(b)}}(x_b)
    =\sum_{i_{e^*}}\sum_{j_a=1,j_b=1}^{c_a,c_b}u^{(j_a)}(x_a)\cdot w^{(j_b)}(x_b)\\
    &\cdot\left(\prod_{e\in\inc(a)} u^{(j_a)}_e(i_e)\right)\cdot\left(\prod_{e\in\inc(b)} w^{(j_b)}_e(i_e)\right)\\
    &=\sum_{j_a=1,j_b=1}^{c_a,c_b}u^{(j_a)}(x_a)\cdot w^{(j_b)}(x_b)\cdot\left(\prod_{e\in\inc(a)\setminus\{e^*\}} u^{(j_a)}_e(i_e)\right)\cdot\left(\prod_{e\in\inc(b)\setminus\{e^*\}} w^{(j_b)}_e(i_e)\right)\\
    &\cdot\sum_{i_{e^*}}\left(u^{(j_a)}_{e^*}(i_{e^*})\cdot w^{(j_b)}_{e^*}(i_{e^*})\right)\\
    &=\sum_{j_a=1,j_b=1}^{c_a,c_b}u^{(j_a)}(x_a)\cdot w^{(j_b)}(x_b)\cdot\left(\prod_{e\in\inc(a)\setminus\{e^*\}} u^{(j_a)}_e(i_e)\right)\cdot\left(\prod_{e\in\inc(b)\setminus\{e^*\}} w^{(j_b)}_e(i_e)\right)\\
    &\left((\mathbf{u}_{e^*}^{(j_a)})^T\cdot \mathbf{w}^{(j_b)}_{e^*}\right),
\end{align*}
where we wrote the vectors at the end of the last expression in bold to emphasise that they are vectors, and part of a dot product.

Applying the gauge transformation only acts on the $u_{e^*}$ and $w_{e^*}$ components, so the only change in the value of this expression will be that the last bracket becomes $\left((X\mathbf{u}_{e^*}^{(j_a)})^T\cdot ((X^T)^{-1}\mathbf{w}^{(j_b)}_{e^*})\right)$, but this doesn't change its value. Therefore, $$\sum_{i_{e^*}}T^{(a)}_{\{i_{e}\}_{e\in \inc(a)}}(x_a)\cdot T^{(b)}_{\{i_e\}_{e\in\inc(b)}}(x_b)$$ doesn't change its value, and so 
\begin{align*}
    \sum_{\{i_e\}_{e\in E}}\prod_{v\in V} T^{(v)}_{\{i_e\}_{e\in \inc(v)}}(x_v)\
    &=\sum_{\{i_e\}_{e\in E\setminus\{e^*\}}}\prod_{v\in V\setminus\{a,b\}}T^{(v)}_{\{i_e\}_{e\in \inc(v)}}(x_v)\\
    &\cdot \left(\sum_{i_{e^*}}T^{(a)}_{\{i_{e}\}_{e\in \inc(a)}}(x_a) \cdot T^{(b)}_{\{i_e\}_{e\in\inc(b)}}(x_b)\right)
\end{align*}
remains unchanged as well, so it will equal $\mathcal{T}(x_V)$ after the gauge transformation.
\end{proof}

\subsection{Proof of the Nonzero Sums Assumption} \label{sec:nonzero_assumption}
In this section we will prove the Nonzero Sums Lemma which will be used in the proof of Theorem \ref{theorem:TCS_reals} in the next section of this appendix. We start with another result which will be used in the proof of the lemma.

\begin{definition}[$\exists\forall$-Partitions]
    Let $A\in\mathbb{R}^{|\mathcal{X}_a|}\otimes\left(\bigotimes_{e\in E_a}\mathbb{R}^{r_e}\right)$ and $B\in\mathbb{R}^{|\mathcal{X}_b|}\otimes\left(\bigotimes_{e\in E_b}\mathbb{R}^{r_e}\right)$ be two tensors with exactly one common mode $\{e^*\}=E_a\cap E_b$. Given partition $E_a=E_a^\exists\dot\cup E_a^\forall$ for $E_a$, then we say $(E_a^\exists,E_a^\forall)$ is a \emph{$\exists\forall$-partition} for $A$ if the following property holds:
    \begin{gather*}
        \exists\{1\leq i_e\leq r_e\}_{e\in E_a^\exists}\text{ such that } \forall\{1\leq i_e\leq r_e\}_{e\in E_a^\forall},\\
        \sum_{x_a} A_{\{i_e\}_{e\in E_a}}(x_a)\neq 0.
    \end{gather*}
    
    Given partitions $E_a=E_a^\exists\dot\cup E_a^\forall$ and $E_b=E_b^\exists\dot\cup E_b^\forall$ for $E_a$ and $E_b$, then we say the pair of partitions $(E_a^\exists,E_a^\forall;E_b^\exists,E_b^\forall)$ is a \emph{$\exists\forall$-partition pair} if $(E_a^\exists,E_a^\forall)$ is a $\exists\forall$-partition for $A$ and $(E_b^\exists,E_b^\forall)$ is a $\exists\forall$-partition for $B$.
\end{definition}

\begin{lemma}   \label{lemma_one_sided_nonzero_sums}
    Let $A\in\mathbb{R}^{|\mathcal{X}_a|}\otimes\left(\bigotimes_{e\in E_a}\mathbb{R}^{r_e}\right)$ and $B\in\mathbb{R}^{|\mathcal{X}_b|}\otimes\left(\bigotimes_{e\in E_b}\mathbb{R}^{r_e}\right)$ be two tensors with exactly one common mode $\{e^*\}=E_a\cap E_b$ and a $\exists\forall$-partition pair $(E_a^\exists,E_a^\forall;E_b^\exists,E_b^\forall)$. Then there exists invertible $r_{e^*}\times r_{e^*}$ matrix $X$ with real entries such that $$(E_a^\exists\setminus\{e^*\},E_a^\forall\cup\{e^*\};E_b^\exists,E_b^\forall)$$ is a $\exists\forall$-partition pair for tensors $X\times_{i_{e^*}}A$ and $(X^T)^{-1}\times_{i_{e^*}}B$.
\end{lemma}
\begin{proof}
If $e^*\in E_a^\forall$ then we are done (pick $X$ as the identity). Assume therefore that $e^*\in E_a^\exists$. Let $\{1\leq i_e\leq r_e\}_{e\in E_a^\exists}$ be an assignment to the indices of $A$ in $E_a^\exists$ such that 
    \begin{equation}    \label{eq:lemma_the_assignment}
        \forall\{1\leq i_e\leq r_e\}_{e\in E_a^\forall},
        \quad \quad \sum_{x_a} A_{\{i_e\}_{e\in E_a}}(x_a)\neq 0.
    \end{equation}\\
Denote by $q$ the value of $i_{e^*}$ in this assignment.

We define $X$ to have the following structure:
$$X=\begin{pmatrix}
        1 & 0 & \dots & 0\\
        0 & 1 & \dots & 0\\
          &   &\ddots &  \\
        0 & 0 & \dots & 1
    \end{pmatrix}
    +
    \begin{pmatrix}
    0 & \dots & \alpha_1 & \dots & 0\\
    0 & \dots & \alpha_2 & \dots & 0\\
      &       & \vdots   &       &  \\
    0 & \dots & \alpha_{r_{e^*}} & \dots & 0
    \end{pmatrix},$$
i.e. it has $1$s on the main diagonal and some real coefficients $\alpha$ which we will set later on column $q$. The only coefficient we set for now is $\alpha_q:=0$, so the row corresponding to $i_{e^*}=q$ only contains a single $1$ and no additional $\alpha$ entry.

Then one can check that by calculating $X\times_{e^*} A$ we get $$\left(X\times_{e^*}A\right)_{\{i_e\}_{e\in E_a}}(x_a) =\begin{cases}
    A_{\{i_e\}_{e\in E_a}}(x_a)&,\text{ if } i_{e^*}=q\\
    A_{\{i_e\}_{e\in E_a}}(x_a) + \alpha_{i_{e^*}}\left[A_{\{i_e\}_{e\in E_a}}(x_a)\right]_{i_{e^*}=q}&,\text{ if } i_{e^*}\neq q,
\end{cases}$$
where by $\left[A_{\{i_e\}_{e\in E_a}}(x_a)\right]_{i_{e^*}=q}$ we denote the expression inside the square brackets with every occurrence of $i_{e^*}$ replaced by the value $q$.

Therefore, if we fix $\{1\leq i_e\leq r_e\}_{e\in E_a^\exists\setminus\{e^*\}}$ to have the values we considered before (i.e. the ones satisfying \eqref{eq:lemma_the_assignment}), the condition
$$\forall\{1\leq i_e\leq r_e\}_{e\in E_a^\forall\cup\{e^*\}},
        \quad \quad \sum_{x_a} \left(X\times_{e^*}A\right)_{\{i_e\}_{e\in E_a}}(x_a)\neq 0$$
is equivalent to
\begin{equation*}
    \forall\{1\leq i_e\leq r_e\}_{e\in E_a^\forall\cup\{e^*\}},
        \quad \quad \sum_{x_a} A_{\{i_e\}_{e\in E_a}}(x_a) + \alpha_{i_{e^*}}\left[\sum_{x_a}A_{\{i_e\}_{e\in E_a}}(x_a)\right]_{i_{e^*}=q}\neq 0.
\end{equation*}
Note that since $\alpha_q=0$, this statement does take into consideration the case $i_{e^*}=q$ as well.

Now since we know that $$\left[\sum_{x_a}A_{\{i_e\}_{e\in E_a}}(x_a)\right]_{i_{e^*}=q}\neq 0$$ must hold because of the way $\{i_e\}_{e\in E_a^\exists\setminus\{e^*\}}$ and $q$ were chosen, we can rewrite the above condition in the equivalent form
\begin{equation}    
    \forall\{1\leq i_e\leq r_e\}_{e\in E_a^\forall\cup\{e^*\}},
        \quad \quad
        \alpha_{i_{e^*}}\neq -\frac{\sum_{x_a} A_{\{i_e\}_{e\in E_a}}(x_a)}{\left[\sum_{x_a}A_{\{i_e\}_{e\in E_a}}(x_a)\right]_{i_{e^*}=q}},\label{eq:appendix_a_side}
\end{equation}
so these are the constraints that the $\alpha$ coefficients need to satisfy in order for $E_a^\exists\setminus\{e^*\}$ and $E_a^\forall\cup\{e^*\}$ to form a $\exists\forall$-partition for $X\times_{e^*}A$.

Now let's look at $(X^T)^{-1}\times_{e^*} B$. We have that $$(X^T)^{-1}=
\begin{pmatrix}
    1 & 0 &  & \dots &  & 0\\
    0 & 1 &  & \dots &  & 0\\
      &   &  & \vdots&  &  \\
    -\alpha_1 & -\alpha_2 & \dots & 1 & \dots & -\alpha_{r_{e^*}}\\
      &   &  & \vdots&  &  \\
    0 & 0 &  & \dots &  & 1
\end{pmatrix}.
$$ Then one can calculate that
\begin{equation}    \label{eq:lemma_B_after_transformation}
\left((X^T)^{-1}\times_{e^*} B\right)_{\{i_e\}_{e\in E_b}}(x_b) =\begin{cases}
    B_{\{i_e\}_{e\in E_b}}(x_b) - \sum_{q'=1}^{r_{e^*}} \alpha_{q'}\left[B_{\{i_e\}_{e\in E_b}}(x_b)\right]_{i_{e^*}=q'}&,\text{ if } i_{e^*}= q\\
    B_{\{i_e\}_{e\in E_b}}(x_b)&,\text{ if } i_{e^*}\neq q.
\end{cases}
\end{equation}
Now pick some values $\{1\leq i_e\leq r_e\}_{e\in E_b^\exists}$ such that 
\begin{equation*}
        \forall\{1\leq i_e\leq r_e\}_{e\in E_b^\forall},
        \quad \quad \sum_{x_b} B_{\{i_e\}_{e\in E_b}}(x_b)\neq 0,
    \end{equation*}
and distinguish two cases.

$\bullet$ \emph{Case $1$:} $e^*\in E_b^\exists.$ If we are in this case and there additionally exists $i_{e^*}=q'\neq q$ such that
\begin{equation*}
        \forall\{1\leq i_e\leq r_e\}_{e\in E_b^\forall},
        \quad \quad \left[\sum_{x_b} B_{\{i_e\}_{e\in E_b}}(x_b)\right]_{i_{e^*}=q'}\neq 0,
    \end{equation*}
then the same is true for $(X^T)^{-1}\times_{e^*}B$ for any values for the coefficients $\alpha$, by equation \eqref{eq:lemma_B_after_transformation}. Pick these coefficients to satisfy \eqref{eq:appendix_a_side} and the matrix $X$ obtained in this way is exactly what we need.

Otherwise, $i_{e^*}=q$ is the only value for which 
\begin{equation*}
        \forall\{1\leq i_e\leq r_e\}_{e\in E_b^\forall},
        \quad \quad \left[\sum_{x_b} B_{\{i_e\}_{e\in E_b}}(x_b)\right]_{i_{e^*}=q}\neq 0
    \end{equation*}
holds. In this case, the condition $$\forall\{1\leq i_e\leq r_e\}_{e\in E_b^\forall},
        \quad \quad \left[\sum_{x_b} \left((X^T)^{-1}\times_{e^*}B\right)_{\{i_e\}_{e\in E_B}}(x_b)\right]_{i_{e^*}=q}\neq 0$$
is equivalent to
\begin{equation}   
    \forall\{1\leq i_e\leq r_e\}_{e\in E_b^\forall},
        \quad \quad \left[\sum_{x_b} B_{\{i_e\}_{e\in E_b}}(x_b)\right]_{i_{e^*}=q} - \sum_{q'=1}^{r_{e^*}} \alpha_{q'}\left[\sum_{x_b} B_{\{i_e\}_{e\in E_b}}(x_b)\right]_{i_{e^*}=q'}\neq 0.   \label{eq:appendix_b_side}
\end{equation}

We view these as inequations in the $\alpha$ coefficients. Now we know that $$\left[\sum_{x_b} B_{\{i_e\}_{e\in E_b}}(x_b)\right]_{i_{e^*}=q}\neq 0$$ holds in all of these equations because of the way $\{i_e\}_{e\in E_b^\exists\setminus\{e^*\}}$ and $q$ were chosen. Therefore, for each inequation in \eqref{eq:appendix_a_side} and \eqref{eq:appendix_b_side}, the set of vectors $\boldsymbol{\alpha}$ which do not satisfy said inequation is a hyperplane in the space of the $\boldsymbol{\alpha}$s (or empty). As a consequence, the set of $\boldsymbol{\alpha}$s which satisfy all of these inequations is nonempty (and in fact infinite), since it is a real vector space with a finite number of hyperplanes removed. We now pick any of these solutions to construct our matrix $X$ and we are done.

$\bullet$ \emph{Case $2$:} $e^*\in E_b^\forall.$ For any $i_{e^*}=q'\neq q$, the condition
$$\forall\{1\leq i_e\leq r_e\}_{e\in E_b^\forall},
        \quad \quad \left[\sum_{x_b} B_{\{i_e\}_{e\in E_B}}(x_b)\right]_{i_{e^*}=q'}\neq 0$$
remains true after applying $(X^T)^{-1}$ as well, since the values of these sums do not change by \eqref{eq:lemma_B_after_transformation}.

We just need to pick the $\boldsymbol{\alpha}$ vector in such a way that 
\begin{equation*}
    \forall\{1\leq i_e\leq r_e\}_{e\in E_b^\forall\setminus\{e^*\}},
        \quad \quad \left[\sum_{x_b} B_{\{i_e\}_{e\in E_b}}(x_b)\right]_{i_{e^*}=q} - \sum_{q'=1}^{r_{e^*}} \alpha_{q'}\left[\sum_{x_b} B_{\{i_e\}_{e\in E_b}}(x_b)\right]_{i_{e^*}=q'}\neq 0
\end{equation*}
holds. However, these equations are the same as \eqref{eq:appendix_b_side}, so by applying the same reasoning as in Case $1$, we can again find a matrix $X$ that satisfies our requirements.
\end{proof}

\begin{corollary}   \label{corollary_nonzero_sums}
    Let $A\in\mathbb{R}^{|\mathcal{X}_a|}\otimes\left(\bigotimes_{e\in E_a}\mathbb{R}^{r_e}\right)$ and $B\in\mathbb{R}^{|\mathcal{X}_b|}\otimes\left(\bigotimes_{e\in E_b}\mathbb{R}^{r_e}\right)$ be two tensors with exactly one common mode $\{e^*\}=E_a\cap E_b$ and a $\exists\forall$-partition pair $(E_a^\exists,E_a^\forall;E_b^\exists,E_b^\forall)$. Then there exists invertible $r_{e^*}\times r_{e^*}$ matrix $X$ with real entries such that $$(E_a^\exists\setminus\{e^*\},E_a^\forall\cup\{e^*\};E_b^\exists\setminus\{e^*\},E_b^\forall\cup\{e^*\})$$ is a $\exists\forall$-partition pair for tensors $X\times_{e^*}A$ and $(X^T)^{-1}\times_{e^*}B$.
\end{corollary}
\begin{proof}
    Apply Lemma \ref{lemma_one_sided_nonzero_sums} to find a matrix $Y$ such that $$(E_a^\exists\setminus\{e^*\},E_a^\forall\cup\{e^*\};E_b^\exists,E_b^\forall)$$ is a $\exists\forall$-partition pair for tensors $Y\times_{e^*}A$ and $(Y^T)^{-1}\times_{e^*}B$. Now apply Lemma \ref{lemma_one_sided_nonzero_sums} again for $(A;B):=\left((Y^T)^{-1}\times_{e^*} B;\,Y\times_{e^*}A\right)$ to find $Z$ such that $$(E_a^\exists\setminus\{e^*\},E_a^\forall\cup\{e^*\};E_b^\exists\setminus\{e^*\},E_b^\forall\cup\{e^*\})$$ is a $\exists\forall$-partition pair for tensors $Z\times_{e^*}\left(Y\times_{e^*}A\right)$ and $(Z^T)^{-1}\times_{e^*}\left((Y^T)^{-1}\times_{e^*}B\right)$. But $Z\times_{e^*}\left(Y\times_{e^*}A\right)=(ZY)\times_{e^*}A$ and $(Z^T)^{-1}\times_{e^*}\left((Y^T)^{-1}\times_{e^*}B\right)=((ZY)^T)^{-1}\times_{e^*} B$, so the matrix $X:=ZY$ satisfies our requirements.
\end{proof}

We are now ready to prove the Nonzero Sums Lemma.
\begin{lemma}[The Nonzero Sums Lemma]   \label{lemma_nonzero_sums}
    Let $\mathcal{G}$ be a tensor network with vector spaces for its nodes $\{\mathbb{R}^{|\mathcal{X}_v|}\}_{v\in V}$ and bond dimensions $\{r_e\}_{e\in E}$, and let $\mathcal{T}\in\bigotimes_{v\in V}\mathbb{R}^{|\mathcal{X}_v|}$ be a tensor which decomposes with respect to the above construction. Assume $\sum_{x_V}\mathcal{T}(x_V)\neq 0$. Then there exists a tensor network decomposition for $\mathcal{T}$ given by tensors $\left\{T^{(v)}\in \mathbb{R}^{|\mathcal{X}_v|}\otimes\left(\bigotimes_{e\in\inc(v)} \mathbb{R}^{r_e}\right)\right\}_{v\in V}$ with the property that $$\sum_{x_v}T^{(v)}_{\{i_e\}_{e\in\inc(v)}}(x_v)\neq 0$$ holds for any $v\in V$ and any values $\{1\leq i_e\leq r_e\}_{e\in\inc(v)}$ for its adjacent indices.
\end{lemma}
\begin{proof}
    Let $\left\{T^{(v)}\in \mathbb{R}^{|\mathcal{X}_v|}\otimes\left(\bigotimes_{e\in\inc(v)} \mathbb{R}^{r_e}\right)\right\}_{v\in V}$ be an arbitrary decomposition of $\mathcal{T}$ with respect to $\mathcal{G}$ and the given bond dimensions. In other words, we have the following equation:
    $$\mathcal{T}(x_V)=\sum_{\{1\leq i_e\leq r_e\}_{e\in E}}\prod_{v\in V} T^{(v)}_{\{i_e\}_{e\in\inc(v)}}(x_v).$$
    Summing over $x_V$, we get:
    $$\sum_{x_V}\mathcal{T}(x_V)=\sum_{\{1\leq i_e\leq r_e\}_{e\in E}}\prod_{v\in V}\left(\sum_{x_v} T^{(v)}_{\{i_e\}_{e\in\inc(v)}}(x_v)\right).$$
    Since $\sum_{x_V}\mathcal{T}(x_V)\neq 0$ by hypothesis, we can conclude that for any vertex $v\in V$, there exist $\{1\leq i_e\leq r_e\}_{e\in\inc(v)}$ values for its adjacent indices such that $$\sum_{x_v} T^{(v)}_{\{i_e\}_{e\in\inc(v)}}(x_v)\neq 0.$$
    That is to say, $(\inc(v),\emptyset)$ is a $\exists\forall$-partition for $T^{(v)}$ for any vertex $v$.

    Finally our proof of this lemma involves an inductive approach. Start with our tensor decomposition for $\mathcal{T}$ and with the initial $\exists\forall$-partition for every $T^{(v)}$ given by $(\inc(v),\emptyset)$. Now pick any edge $e\in E$ and apply Corollary \ref{corollary_nonzero_sums} to find a gauge transformation which makes it so that $e$ will be in the $\forall$-component of both of its adjacent $T^{(v)}$ tensors afterwards. We apply this gauge transformation, replace our decomposition with the resulting decomposition, and modify the $\exists\forall$-partitions of the vertices adjacent to $e$ so that $e$ is in their $\forall$-components. We repeatedly apply this process until every edge in the graph is in the $\forall$-component of their adjacent tensors. At the end, each $T^{(v)}$ will have $\exists\forall$-partition $(\emptyset,\inc(v))$, and this means exactly that $$\forall v\in V,\{1\leq i_e\leq r_e\}_{e\in\inc(v)},\quad\quad\sum_{x_v}T^{(v)}_{\{i_e\}_{e\in\inc(v)}}(x_v)\neq 0.$$
\end{proof}

\subsection{Proof of Theorem \ref{theorem:TCS_reals}}
\label{sec:actual_conjecture_proof}
\begin{no_number_theorem}
    For any Tree-structured Correlation Scenario $\mathcal{G}$, a distribution is in the quasimarginal model if and only if it satisfies the observed independence constraints enforced by the graph. In other words, $$\mathcal{W}(\mathcal{G})=\mathcal{M}(\mathcal{G})=\mathcal{N}(\mathcal{G}).$$
\end{no_number_theorem}

\begin{proof}
We recommend the reader familiarises themselves with the simpler proof in the main text first. We will prove the inclusion $\mathcal{M}(\mathcal{G})\subseteq\mathcal{W}(\mathcal{G})$, since the rest follows from Propositions \ref{proposition:quasi_is_in_nested} and \ref{proposition:correlation_scenario_properties}.

As before, consider a minimal tensor decomposition of $p\in\mathcal{M}(\mathcal{G})$ with respect to $\mathcal{G}^t$. Let $r_e$ be the bond dimension of edge $e\in E^t$ of this decomposition. Recall that the tensor decomposition of $p$ implies the existence of $\left\{T^{(v)}\in \mathbb{R}^{|\mathcal{X}_v|}\otimes\left(\bigotimes_{e\in\inc(v)} \mathbb{R}^{r_e}\right)\right\}_{v\in V}$ such that
\begin{equation}    \label{eq:appendix_tensor_decomposition}
    p(x_V)=\sum_{\{1\leq i_e\leq r_e\}_{e\in E}}\prod_{v\in V} T^{(v)}_{\{i_e\}_{e\in\inc(v)}}(x_v).
\end{equation}

We can assume by Lemma \ref{lemma_nonzero_sums} in Appendix \ref{sec:nonzero_assumption} that $$\sum_{x_V} T^{(v)}_{\{i_e\}_{e\in\inc(v)}}(x_v)\neq0$$ holds for all $v\in V$ and values for indices $\{i_e\}_{e\in\inc(v)}$.

We can therefore write
\begin{equation}    \label{eq:appendix_factorisation_before_rank_1}
    p(x_V)=\sum_{\{1\leq i_e\leq r_e\}_{e\in E}}\prod_{v\in V} \frac{T^{(v)}_{\{i_e\}_{e\in\inc(v)}}(x_v)}{\sum_{x_v}T^{(v)}_{\{i_e\}_{e\in\inc(v)}}(x_v)}\cdot\prod_{v\in V}\left(\sum_{x_v}T^{(v)}_{\{i_e\}_{e\in\inc(v)}}(x_v)\right).
\end{equation}
Since $\sum_{x_v}\frac{T^{(v)}_{\{i_e\}_{e\in\inc(v)}}(x_v)}{\sum_{x_v}T^{(v)}_{\{i_e\}_{e\in\inc(v)}}(x_v)}=1$, we can create for each edge $e$ of the graph $\mathcal{G}^t$ the corresponding latent variable $l_e$ in $\mathcal{G}$ with size $r_e$ and the response functions of the observed variables to be $$p\left(x_v|\{x_{l_e}=i_e\}_{e\in \inc(v)}\right):=\frac{T^{(v)}_{\{i_e\}_{e\in\inc(v)}}(x_v)}{\sum_{x_v}T^{(v)}_{\{i_e\}_{e\in\inc(v)}}(x_v)}.$$

We just need to set $p(x_{l_e}=i_e)$ for each edge $e\in E^t$ and $1\leq i_e\leq r_e$. To do so, we will prove that for each vertex $v\in V$, the quantity $\sum_{x_v}T^{(v)}_{\{i_e\}_{e\in\inc(v)}}(x_v)$ decomposes into factors depending on only one index as $$\sum_{x_v}T^{(v)}_{\{i_e\}_{e\in\inc(v)}}(x_v)=\prod_{e\in\inc(v)} \phi^{(v)}_e(i_e),$$ for some functions $\{\phi^{(v)}_e\}_{v\in V, e\in\inc(v)}$.

Let $a\in V$ be an arbitrary vertex of $\mathcal{G}^t$. We sum over $x_a$ in equation \eqref{eq:appendix_tensor_decomposition}:
\begin{equation}    \label{eq:tensor_decomposition_summed}
    p(x_{V\setminus \{a\}})=\sum_{\{1\leq i_e\leq r_e\}_{e\in E}}\prod_{v\in V\setminus\{a\}} T^{(v)}_{\{i_e\}_{e\in\inc(v)}}(x_v)\cdot \left(\sum_{x_a}T^{(a)}_{\{i_e\}_{e\in\inc(a)}}(x_a)\right).
\end{equation}
As before, we denote the branches left after eliminating $a$ from the graph by $$\mathcal{B}_a=\{B\subseteq V\setminus\{a\}:B\text{ is a maximal connected component of the graph } \mathcal{G}\setminus\{a\}\}.$$

We now redistribute the entries of the order $|V|-1$ tensor $p(x_{V\setminus\{a\}})$ to obtain another tensor with fewer modes, in particular one mode for each branch of $\mathcal{B}_a$. We do this in the same way we did in the simpler proof in the main text: for each branch $B\in\mathcal{B}_a$, group the modes corresponding to vertices $v_i$ in $B=\{v_1,\dots,v_s\}$ into a single one corresponding to the tuple $(v_1,\dots,v_s)$. In other words, we replace the modes $x_{v_1},\dots,x_{v_s}$ of sizes $|\mathcal{X}_{v_1}|,\dots,|\mathcal{X}_{v_s}|$ by a single mode $x_B$ of size $\prod_{i=1}^s|\mathcal{X}_{v_i}|$. If we write $\mathcal{B}_a=\{B_1,\dots,B_t\}$, then we will denote the tensor described above by $p(x_{B_1};x_{B_2};\dots;x_{B_t})$.

We also do some regrouping in the RHS of equation \eqref{eq:tensor_decomposition_summed}. For each branch $B=\{v_1,\dots,v_s\}\in\mathcal{B}_a$, we group the factors $T^{(v_i)}_{\{i_e\}_{e\in\inc(v_i)}}(x_{v_i})$ together. As such, denote $$M^{(B)}_{i_{e_B}}(x_B):=\sum_{\{1\leq i_e\leq r_e\}_{e\in\inc(B)\setminus{\{e_B\}}}}\prod_{v_i\in B}T^{(v_i)}_{\{i_e\}_{e\in\inc(v_i)}}(x_{v_i}),$$ where $\inc(B)=\bigcup_{v_i\in B}\inc(v_i)$, and $e_B\in E$ is the unique edge that links the eliminated vertex $a$ and some vertex in the branch $B$. Now we can rewrite \eqref{eq:tensor_decomposition_summed} as \begin{equation} \label{eq:tensor_decomposition_summed_partial_rewrite}
    p(x_{V\setminus \{a\}})=\sum_{\{1\leq i_{e_B}\leq r_{e_B}\}_{B\in \mathcal{B}_a}}\left(\prod_{B\in\mathcal{B}_a} M^{(B)}_{i_{e_B}}(x_B)\right)\cdot \left(\sum_{x_a}T^{(a)}_{\{i_e\}_{e\in\inc(a)}}(x_a)\right).
\end{equation}

Let us look at what the object $M^{(B)}$ is in detail for a given branch $B$. One can see it as a tensor with $|B|+1$ modes: one mode for each $v_i\in B$ and an additional mode corresponding to the value of the index $i_{e_B}$. However, similar to the proof in the main text, we choose to flatten it by grouping the modes corresponding to $x_{v_i}$ into a single mode corresponding to the value of the tuple $x_B$. The resulting tensor is a $|\mathcal{X}_B|\times r_{e_B}$ matrix having rows indexed by values of the tuple $x_B$ and columns indexed by $i_{e_B}$. For brevity we will keep the same notation $M^{(B)}$, but keep in mind this refers to a matrix from now on.

Now if we view the expression $\sum_{x_a}T^{(a)}_{\{i_e\}_{e\in\inc(a)}}(x_a)$ as an order $|\mathcal{B}_a|=t$ tensor with a mode $i_{e_B}$ for each branch $B\in\mathcal{B}_a$, then we can rewrite \eqref{eq:tensor_decomposition_summed_partial_rewrite} as:
\begin{equation}    \label{eq:tensor_decomposition_summed_rewritten}
    p(x_{B_1};\dots;x_{B_t})=\left(M^{(B_1)},M^{(B_2)},\dots,M^{(B_t)}\right)\times\left(\sum_{x_a}T^{(a)}_{\{i_e\}_{e\in\inc(a)}}(x_a)\right),
\end{equation}
where $\left(M^{(B_1)},M^{(B_2)},\dots,M^{(B_t)}\right)$ is the multilinear transformation operator defined in Appendix \ref{sec:multilinear_products}.

Recall that $M^{(B)}$ has rows indexed by $x_B$ and columns indexed by $i_{e_B}$. We now claim that $M^{(B)}$ has linearly independent columns for any $B\in\mathcal{B}_a$. Indeed, assume the contrary holds for some $B^*\in\mathcal{B}_a$. Then there exists some $1\leq q\leq r_{e_{B^*}}$ and real numbers $\alpha_{i_{e_{B^*}}}$ such that $$M^{(B^*)}_{q}(x_{B^*})=\sum_{i_{e_{B^*}}\neq q}\alpha_{i_{e_{B^*}}}\cdot M^{(B^*)}_{i_{e_{B^*}}}(x_{B^*})$$ for all $x_{B^*}\in\mathcal{X}_{B^*}$. 

We can replace this in equation \eqref{eq:appendix_tensor_decomposition} to obtain:
\begin{align*}
    p(x_V)
    &=\sum_{\{1\leq i_e\leq r_e\}_{e\in E}}\prod_{v\in V} T^{(v)}_{\{i_e\}_{e\in\inc(v)}}(x_v)\\
    &=\sum_{\{1\leq i_e\leq r_e\}_{e\in E\setminus{\inc(B^*)}}}\left(\prod_{v\in V\setminus({B^*}\cup\{a\})} T^{(v)}_{\{i_e\}_{e\in\inc(v)}}(x_v)\right)\\
    &\cdot \sum_{1\leq i_{e_{B^*}}\leq r_{e_{B^*}}} T^{(a)}_{\{i_e\}_{e\in\inc(a)}}(x_a)\cdot\left(\sum_{\{1\leq i_e\leq r_e\}_{e\in\inc(B^*)\setminus\{e_{B^*}\}}}\prod_{v_i\in B^*} T^{(v_i)}_{\{i_e\}_{e\in\inc(v_i)}}(x_{v_i})\right)\\
    &=\sum_{\{1\leq i_e\leq r_e\}_{e\in E\setminus{\inc(B^*)}}}\left(\prod_{v\in V\setminus({B^*}\cup\{a\})} T^{(v)}_{\{i_e\}_{e\in\inc(v)}}(x_v)\right)\\
    &\cdot \sum_{1\leq i_{e_{B^*}}\leq r_{e_{B^*}}} T^{(a)}_{\{i_e\}_{e\in\inc(a)}}(x_a)\cdot M^{(B^*)}_{i_{e_{B^*}}}(x_{B^*})
\end{align*}

However, if we look at the expression on the last line, one can see that for all values of $\{i_e\}_{e\in E\setminus\inc(B^*)}$:
\begin{align*}
&\sum_{1\leq i_{e_{B^*}}\leq r_{e_{B^*}}} T^{(a)}_{\{i_e\}_{e\in\inc(a)}}(x_a)\cdot M^{(B^*)}_{i_{e_{B^*}}}(x_{B^*})\\
&=\sum_{i_{e_{B^*}}\neq q}T^{(a)}_{\{i_e\}_{e\in\inc(a)}}(x_a)\cdot M^{(B^*)}_{i_{e_{B^*}}}(x_{B^*}) + 
\left[T^{(a)}_{\{i_e\}_{e\in\inc(a)}}(x_a)\right]_{i_{e_{B^*}}=q}\cdot M^{(B^*)}_q(x_{B^*})\\
&=\sum_{i_{e_{B^*}}\neq q}T^{(a)}_{\{i_e\}_{e\in\inc(a)}}(x_a)\cdot M^{(B^*)}_{i_{e_{B^*}}}(x_{B^*})+\left[T^{(a)}_{\{i_e\}_{e\in\inc(a)}}(x_a)\right]_{i_{e_{B^*}}=q}\cdot\sum_{i_{e_{B^*}}\neq q}\alpha_{i_{e_{B^*}}}\cdot M^{(B^*)}_{i_{e_{B^*}}}(x_{B^*})\\
&=\sum_{i_{e_{B^*}}\neq q}\left(T^{(a)}_{\{i_e\}_{e\in\inc(a)}}(x_a)+\alpha_{i_{e_{B^*}}}\cdot\left[T^{(a)}_{\{i_e\}_{e\in\inc(a)}}(x_a)\right]_{i_{e_{B^*}}=q}\right)\cdot M^{(B^*)}_{i_{e_{B^*}}}(x_{B^*}).
\end{align*}

Replacing this in the expression for $p(x_V)$ we obtain:
\begin{align*}
    p(x_V) =& \sum_{\{1\leq i_e\leq r_e\}_{e\in E\setminus\{e_{B^*}\}}}\sum_{\substack{1\leq i_{e_{B^*}}\leq r_{e_{B^*}}\\i_{e_{B^*}}\neq q}} \left(\prod_{v\in V\setminus\{a\}} T^{(v)}_{\{i_e\}_{e\in\inc(v)}}(x_v)\right)\\
    &\cdot \left(T^{(a)}_{\{i_e\}_{e\in\inc(a)}}(x_a)+\alpha_{i_{e_{B^*}}}\cdot\left[T^{(a)}_{\{i_e\}_{e\in\inc(a)}}(x_a)\right]_{i_{e_{B^*}}=q}\right).
\end{align*}

Therefore, we obtain a strictly smaller decomposition for $\mathcal{T}$: by removing the value $i_{e_{B^*}}=q$ and replacing the tensor $T^{(a)}_{\{i_e\}_{e\in\inc(a)}}(x_a)$ by the new tensor in the brackets above, we obtain a tensor network decomposition for $p(x_V)$ which has bond dimension for $e_{B^*}$ equal to $r_{e_{B^*}}-1$, contradicting the minimality of the bond dimensions we have picked.

Therefore, indeed the columns of $M^{(B^*)}$ must be linearly independent. As a result, for any $B\in\mathcal{B}_a$, $M^{(B)}$ must have a left inverse $N^{(B)}$ such that $$N^{(B)}M^{(B)}=\mathcal{I}_{r_{e_B}}.$$ Recall equation \eqref{eq:tensor_decomposition_summed_rewritten}:
$$p(x_{B_1};\dots;x_{B_t})=\left(M^{(B_1)},M^{(B_2)},\dots,M^{(B_t)}\right)\times\left(\sum_{x_a}T^{(a)}_{\{i_e\}_{e\in\inc(a)}}(x_a)\right).$$
Applying $\left(N^{(B_1)},N^{(B_2)},\dots,N^{(B_t)}\right)$ to the left of each side of the equation, we get:
$$\left(N^{(B_1)},N^{(B_2)},\dots,N^{(B_t)}\right)\times p(x_{B_1};\dots;x_{B_t})=\sum_{x_a}T^{(a)}_{\{i_e\}_{e\in\inc(a)}}(x_a).$$

However, by $p(x_V)\in\mathcal{M}(\mathcal{G})$ and Proposition \ref{markov_model_in_TCS} we get $$p(x_{V\setminus \{a\}})=\prod_{B\in\mathcal{B}_a}p(x_B),$$
which we can rewrite in terms of tensor products as:
$$p(x_{B_1};\dots;x_{B_t})=\mathbf{p(x_{B_1})}\otimes\dots\otimes \mathbf{p(x_{B_t})},$$ where we write the factors in the RHS in bold to emphasise they are vectors.

Therefore,
\begin{align*}
    \sum_{x_a}T^{(a)}_{\{i_e\}_{e\in\inc(a)}}(x_a)
    &=\left(N^{(B_1)},N^{(B_2)},\dots,N^{(B_t)}\right)\times \left(\mathbf{p(x_{B_1})}\otimes\dots\otimes \mathbf{p(x_{B_t})}\right)\\
    &=\left(N^{(B_1)}\mathbf{p(x_{B_1})}\right)\otimes\dots\otimes\left(N^{(B_t)}\mathbf{p(x_{B_t})}\right).
\end{align*}

Note that $N^{(B)}$ is a $r_{e_B}\times |\mathcal{X}_B|$ matrix, meaning $N^{(B)}\mathbf{p(x_B)}$ is a vector of size $r_{e_B}$, which is what we would expect. Therefore, by setting $$\phi_{e_B}^{(a)}:=N^{(B)}_{i_{e_B}}\mathbf{p(x_B)},$$ we obtain that indeed, for any vertex $v\in V$, $$\sum_{x_v}T^{(v)}_{\{i_e\}_{e\in\inc(v)}}(x_v)=\prod_{e\in\inc(v)} \phi^{(v)}_e(i_e).$$

To recap, by equation \eqref{eq:appendix_factorisation_before_rank_1}, we have:
\begin{equation*}
    p(x_V)=\sum_{\{1\leq i_e\leq r_e\}_{e\in E}}\prod_{v\in V} \frac{T^{(v)}_{\{i_e\}_{e\in\inc(v)}}(x_v)}{\sum_{x_v}T^{(v)}_{\{i_e\}_{e\in\inc(v)}}(x_v)}\cdot\prod_{v\in V}\left(\sum_{x_v}T^{(v)}_{\{i_e\}_{e\in\inc(v)}}(x_v)\right).
\end{equation*}
If we define $$\psi_e(i_e):=\phi_e^{(v_1)}(i_e)\cdot\phi_e^{(v_2)}(i_e),$$ for any edge $e=(v_1,v_2)$, then by equation \eqref{eq:appendix_factorisation_before_rank_1} we obtain:
\begin{align*}
    p(x_V)
    &=\sum_{\{1\leq i_e\leq r_e\}_{e\in E}}\prod_{v\in V} \frac{T^{(v)}_{\{i_e\}_{e\in\inc(v)}}(x_v)}{\sum_{x_v}T^{(v)}_{\{i_e\}_{e\in\inc(v)}}(x_v)}\cdot\prod_{v\in V}\left(\sum_{x_v}T^{(v)}_{\{i_e\}_{e\in\inc(v)}}(x_v)\right)\\
    &=\sum_{\{1\leq i_e\leq r_e\}_{e\in E}}\prod_{v\in V} p\left(x_v|\{x_{l_e}=i_e\}_{e\in \inc(v)}\right)\cdot \prod_{v\in V}\prod_{e\in\inc(v)}\phi^{(v)}_e(i_e)\\
    &=\sum_{\{1\leq i_e\leq r_e\}_{e\in E}}\prod_{v\in V} p\left(x_v|\{x_{l_e}=i_e\}_{e\in \inc(v)}\right)\cdot \prod_{e=(v_1,v_2)\in E} \phi_e^{(v_1)}(i_e)\cdot\phi_e^{(v_2)}(i_e)\\
    &=\sum_{\{1\leq i_e\leq r_e\}_{e\in E}}\prod_{v\in V} p\left(x_v|\{x_{l_e}=i_e\}_{e\in \inc(v)}\right)\cdot \prod_{e\in E} \psi_e(i_e).
\end{align*}

By summing this equation over $x_V$ we obtain 
$$
    1=\sum_{\{1\leq i_e\leq r_e\}_{e\in E}}\prod_{e\in E} \psi_e(i_e)=\prod_{e\in E}\sum_{1\leq i_e\leq r_e}\psi_e(i_e).
$$

Therefore, 
\begin{equation}    \label{eq:appendix_factorisation_last_form}
    p(x_V)=\sum_{\{1\leq i_e\leq r_e\}_{e\in E}}\prod_{v\in V} p\left(x_v|\{x_{l_e}=i_e\}_{e\in \inc(v)}\right)\cdot \prod_{e\in E} \frac{\psi_e(i_e)}{\sum_{1\leq i_e\leq r_e}\psi_e(i_e)}.
\end{equation}

Finally, we go back to the original correlation scenario $\mathcal{G}$. Define the latent variable $l_e$ corresponding to edge $e$ in the graph $\mathcal{G}^t$ to have domain $\{1,\dots,r_e\}$ and quasiprobability distribution $$p(x_{l_e}=i_e):=\frac{\psi_e(i_e)}{\sum_{1\leq i_e\leq r_e}\psi_e(i_e)}.$$ Then for all latent nodes $l_e$ we have $\sum_{x_{l_e}} p(x_{l_e})=1$, and for all observed nodes $v$ we have $\sum_{x_v}p\left(x_v|\{x_{l_e}=i_e\}_{e\in \inc(v)}\right)=1$ as expected, and equation \eqref{eq:appendix_factorisation_last_form} tells us that $$p(x_V)=\sum_{x_L}\prod_{v\in V} p(x_v|x_{\pa_{\mathcal{G}}(v)})\prod_{l\in L}p(x_l).$$ 

Therefore, $p(x_V)$ is in the quasi set $\mathcal{W}(\mathcal{G})$, so $\mathcal{M}(\mathcal{G})\subseteq\mathcal{W}(\mathcal{G})$, and so $$\mathcal{W}(\mathcal{G})=\mathcal{M}(\mathcal{G})=\mathcal{N}(\mathcal{G}).$$

\end{proof}

\section{Conjecture \ref{conjecture} in the Triangle Scenario}  \label{sec:triangle_proof}

\begin{figure}[t]
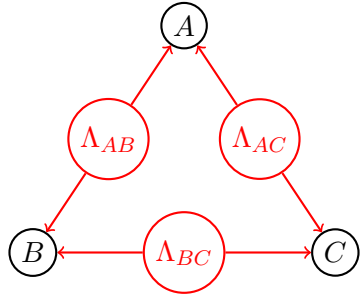

    \ctikzfig{TikzFigures/triangleDAG_letters}
    \caption{The Triangle Scenario.}
    \label{fig:triangle_scenario}
\end{figure}

We prove that Conjecture \ref{conjecture} is true for the Triangle Scenario in Fig. \ref{fig:triangle_scenario}. We start under the assumption that the observed variables are binary, and we will give two constructions of the perfect correlation $p=\frac{1}{2}([000]+[111])$, where $[abc]$ represents the deterministic distribution with $A=a$, $B=b$, and $C=c$. The first one uses quasiprobabilistic latent variables, and the second one is a symmetric construction using quasiprobabilistic response distributions.

The quasilatents construction has $A$, $B$, and $C$ deterministic given their latent parents. We will write the distributions of the latents as vectors, i.e. the distribution $p_{\Lambda_{BC}}$ below has $p_{\Lambda_{BC}}(0)=p_{\Lambda_{BC}}(1)=\frac{1}{2}$:
\begin{itemize}
    \item[$\bullet$] $\Omega_{\Lambda_{BC}}=\{0,1\}$, $p_{\Lambda_{BC}}=(\frac{1}{2},\frac{1}{2})^T$\\
    \item[$\bullet$] $\Omega_{\Lambda_{AC}}=\{1,2,3,4,5,6,7\}$, $p_{\Lambda_{AC}}=(0.1, 0.9, 0.1, 0.3, -0.4, -0.1, 0.1)^T$\\
    \item[$\bullet$] $\Omega_{\Lambda_{AB}}=\{1,2,3,4,5\}$, $p_{\Lambda_{AB}}=(1, 1,-1, 1,-1)^T$\\
    \item[$\bullet$] $A=0\iff (\Lambda_{AC},\Lambda_{AB})\in\{(3,1),(6,1),(1,2),(4,2),(7,2),(1,5),(3,4),(6,4),(7,4)\}$\\
    \item[$\bullet$] $B=0 \iff (\Lambda_{BC}=0 \land \Lambda_{AB}\in\{1,2,3\}) \lor (\Lambda_{BC}=1 \land \Lambda_{AB}\in\{2,3,4,5\})$\\
    \item[$\bullet$] $C=0 \iff (\Lambda_{BC}=0 \land \Lambda_{AC}\in\{1,2,3,4,5\}) \lor (\Lambda_{BC}=1 \land \Lambda_{AC}\in\{3,4,5,6,7\})$.
\end{itemize}

The quasiresponses construction has all latents behaving as unbiased coinflips and the following response for each observed variable:
$$p(v|\lambda_1,\lambda_2)=
\begin{cases}
    \frac{3}{2}\,\,\,\,\,\,\,\,\,\text{ if } \lambda_1=\lambda_2=v;\\
    -\frac{1}{2}\,\,\,\,\text{ if } \lambda_1=\lambda_2\neq v;\\
    \frac{1}{2}\,\,\,\,\,\,\,\,\,\text{ if }\lambda_1\neq\lambda_2,
\end{cases}$$
where $v$ is a placeholder for $a,b,c$, with parents $\lambda_1$ and $\lambda_2$.

Now that the perfect correlation is in $\mathcal{W}$, we can use this construction to show that any distribution is consistent with the triangle scenario.

We start by extending the latent variables in the original scenario with a copy of the perfect correlation construction (more precisely, replace each latent $\Lambda$ by $(\Lambda,\Lambda')$, where $\Lambda'$ is a copy of its corresponding latent variable in either of the constructions above). This results in the observed variables also sharing an unbiased coinflip $U$ (see Fig. \ref{fig:triangle_after_extension}).

\begin{figure}[t]
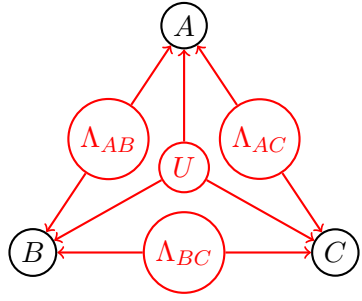

    \ctikzfig{TikzFigures/triangleDAG_after_extension}
    \caption{The Triangle Scenario after extending the latents with a copy of the perfect correlation construction.}
    \label{fig:triangle_after_extension}
\end{figure}
\begin{figure}[t]
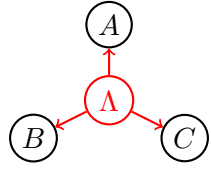

    \ctikzfig{TikzFigures/Common_Ancestor_letters}
    \caption{The Common Ancestor Scenario.}
    \label{fig:common_ancestor}
\end{figure}

Now simply note there is no reason to stop at one copy of $U$. Since the latents are allowed to be continuous, we can instead make them generate an infinite sequence of i.i.d. copies of $U$. Since any finite random variable can be simulated via an infinite number of coinflips, we conclude that using negative probabilities, the Triangle Scenario can simulate the Common Ancestor Scenario in Fig. \ref{fig:common_ancestor}. It is obvious any distribution on three variables is in the Common Ancestor Scenario, so we are done.


\end{document}